\documentclass[letterpaper, onecolumn]{IEEEtran}

\usepackage[mathscr]{eucal}
\usepackage[cmex10]{amsmath}
\usepackage{epsfig,epsf,psfrag}
\usepackage{amssymb,amsmath,amsthm,amsfonts,latexsym}
\usepackage{amsmath,graphicx,bm,xcolor,url}
\usepackage[caption=false]{subfig} 
\usepackage{fixltx2e}
\usepackage{array}
\usepackage{verbatim}
\usepackage{bm}
\usepackage{verbatim}
\usepackage{textcomp}
\usepackage{mathrsfs}
\usepackage{epstopdf}
\usepackage{bbm}

\catcode`~=11 \def\UrlSpecials{\do\~{\kern -.15em\lower .7ex\hbox{~}\kern .04em}} \catcode`~=13 

\allowdisplaybreaks[1]
 
\newcommand{\nn}{\nonumber}

\newcommand{\calA}{\mathcal{A}}
\newcommand{\calB}{\mathcal{B}}

\newcommand{\calE}{\mathcal{E}}

\newcommand{\calG}{\mathcal{G}}

\newcommand{\calM}{\mathcal{M}}

\newcommand{\calP}{\mathcal{P}}

\newcommand{\calR}{\mathcal{R}}

\newcommand{\calT}{\mathcal{T}}
\newcommand{\calU}{\mathcal{U}}

\newcommand{\ba}{\mathbf{a}}

\newcommand{\be}{\mathbf{e}}

\newcommand{\rmc}{\mathrm{c}}

\newcommand{\rmd}{\mathrm{d}}

\newcommand{\rme}{\mathrm{e}}

\newcommand{\rmP}{\mathrm{P}}

\newcommand{\bbE}{\mathbb{E}}
\newcommand{\bbF}{\mathbb{F}}

\newcommand{\bbN}{\mathbb{N}}

\newcommand{\bbR}{\mathbb{R}}

\DeclareMathAlphabet{\mathbsf}{OT1}{cmss}{bx}{n}
\DeclareMathAlphabet{\mathssf}{OT1}{cmss}{m}{sl}

\DeclareSymbolFont{bsfletters}{OT1}{cmss}{bx}{n}  
\DeclareSymbolFont{ssfletters}{OT1}{cmss}{m}{n}
\DeclareMathSymbol{\bsfGamma}{0}{bsfletters}{'000}
\DeclareMathSymbol{\ssfGamma}{0}{ssfletters}{'000}
\DeclareMathSymbol{\bsfDelta}{0}{bsfletters}{'001}
\DeclareMathSymbol{\ssfDelta}{0}{ssfletters}{'001}
\DeclareMathSymbol{\bsfTheta}{0}{bsfletters}{'002}
\DeclareMathSymbol{\ssfTheta}{0}{ssfletters}{'002}
\DeclareMathSymbol{\bsfLambda}{0}{bsfletters}{'003}
\DeclareMathSymbol{\ssfLambda}{0}{ssfletters}{'003}
\DeclareMathSymbol{\bsfXi}{0}{bsfletters}{'004}
\DeclareMathSymbol{\ssfXi}{0}{ssfletters}{'004}
\DeclareMathSymbol{\bsfPi}{0}{bsfletters}{'005}
\DeclareMathSymbol{\ssfPi}{0}{ssfletters}{'005}
\DeclareMathSymbol{\bsfSigma}{0}{bsfletters}{'006}
\DeclareMathSymbol{\ssfSigma}{0}{ssfletters}{'006}
\DeclareMathSymbol{\bsfUpsilon}{0}{bsfletters}{'007}
\DeclareMathSymbol{\ssfUpsilon}{0}{ssfletters}{'007}
\DeclareMathSymbol{\bsfPhi}{0}{bsfletters}{'010}
\DeclareMathSymbol{\ssfPhi}{0}{ssfletters}{'010}
\DeclareMathSymbol{\bsfPsi}{0}{bsfletters}{'011}
\DeclareMathSymbol{\ssfPsi}{0}{ssfletters}{'011}
\DeclareMathSymbol{\bsfOmega}{0}{bsfletters}{'012}
\DeclareMathSymbol{\ssfOmega}{0}{ssfletters}{'012}

\newcommand{\hatA}{\hat{A}}
\newcommand{\tila}{\tilde{a}}

\newcommand{\tilH}{\tilde{H}}

\newcommand{\hatR}{\hat{R}}

\newcommand{\barf}{\overline{f}}

\def\fndot{\, \cdot \,}

\newcommand{\dotleq}{\stackrel{.}{\leq}}

\newcommand{\dotgeq}{\stackrel{.}{\geq}}

\DeclareMathOperator*{\argmax}{arg\,max}

\newcommand{\bone}{\mathbbm{1}}

\newtheorem{theorem}{Theorem} 
\newtheorem{lemma}{Lemma}

\newtheorem{proposition}[theorem]{Proposition}

\newtheorem{definition}{Definition}

\newtheorem{remark}{Remark}

\newcommand{\qednew}{\nobreak \ifvmode \relax \else
      \ifdim\lastskip<1.5em \hskip-\lastskip
      \hskip1.5em plus0em minus0.5em \fi \nobreak
      \vrule height0.75em width0.5em depth0.25em\fi}

\usepackage[ colorlinks = true,
             linkcolor = blue,
             urlcolor  = blue,
             citecolor = red,
             anchorcolor = green,
]{hyperref}

\newcommand{\tba}{\tilde{\ba}}
\usepackage{cite}

\begin{document}
\flushbottom
\allowdisplaybreaks[1]

\title{Analysis of Remaining Uncertainties and Exponents under Various Conditional  R\'enyi Entropies} 

\author{Vincent Y.~F.\ Tan$^\dagger$,~\IEEEmembership{Senior Member,~IEEE},  $\,$ and  $\,$ Masahito Hayashi$^\ddagger$,~\IEEEmembership{Senior Member,~IEEE}
\thanks{$\dagger$Vincent~Y.~F. Tan is with the Department of Electrical and Computer Engineering and the Department of Mathematics, National University of Singapore (Email: \url{vtan@nus.edu.sg}). }
\thanks{$^\ddagger$Masahito~Hayashi is with the  Graduate School of Mathematics, Nagoya University, and the Center for Quantum Technologies (CQT),  National University of Singapore   (Email: \url{masahito@math.nagoya-u.ac.jp}).   }   \thanks{This paper was presented in part at the 2016 International Symposium on Information Theory  (ISIT) in Barcelona, Spain.}   }


\maketitle

\begin{abstract}
In this paper, we analyze the asymptotics of the normalized remaining uncertainty of a source  when a compressed or hashed version of it and correlated side-information   is observed. For this system,   commonly known as Slepian-Wolf  source coding, we establish the optimal (minimum) rate  of compression of the source to ensure that the remaining uncertainties vanish.   We also study the exponential rate of decay of the remaining uncertainty to zero when the rate is above the optimal rate of compression. In our study, we consider various classes of random  universal hash functions. Instead of measuring remaining uncertainties using traditional Shannon information measures, we do so using two forms of the conditional R\'enyi entropy. Among other   techniques, we employ new one-shot bounds and the moments of type class enumerator method for these evaluations.  We show that these asymptotic results are generalizations of  the strong converse exponent and the error exponent of the Slepian-Wolf  problem   under maximum {\em a posteriori} (MAP) decoding.   
\end{abstract}  


\section{Introduction} \label{sec:intro}
In information-theoretic security~\cite{Liang,BlochBarros}, it is of fundamental importance to study the remaining uncertainty of a random variable $A^n$ given a compressed version of itself $f(A^n)$ and another correlated signal $E^n$. This model, reminiscent of the the Slepian-Wolf  source  coding  problem\footnote{In this paper, we abuse terminology and use the terms {\em Slepian-Wolf  coding}~\cite{sw73} and {\em lossless source coding with decoder side-information}  interchangeably.}~\cite{sw73}, is illustrated in Fig.~\ref{fig:sw}. A model somewhat similar to the one we study here was studied by  Tandon, Ulukus and Ramachandran~\cite{tandon13} who analyzed the problem of secure source coding with a helper.  In particular, a party  would like to reconstruct a source $A^n$ given a ``helper'' signal  (or a compressed version of it) but an eavesdropper, who can tap on $f(A^n)$ is also present in the system. The authors in~\cite{tandon13} analyzed the tradeoff between the compression rate and the equivocation of $A^n$ given $f(A^n)$. Villard and   Piantanida~\cite{villard13} and Bross~\cite{bross16} considered the setting in which the eavesdropper also has access to memoryless side-information $E^n$
that is correlated with $A^n$. However, there are many ways that one could measure the equivocation or remaining uncertainty. The traditional way, starting from Wyner's seminal paper on the wiretap channel~\cite{Wyn75} (and also in~\cite{Liang,BlochBarros,tandon13,villard13, bross16}), is to do so using the conditional  Shannon  entropy $H( A^n|f(A^n),E^n)$,  leading to a ``standard'' equivocation measure. In this paper, we   study the asymptotics  of remaining uncertainties based on the family of  R\'enyi  information measures~\cite{renyi}. The measures we consider include the  {\em conditional R\'enyi entropy} $H_{1 + s} (A^n | f(A^n),E^n)$ and its so-called {\em Gallager form}, which we denote as $H_{1 + s}^\uparrow (A^n | f(A^n),E^n)$. We note that unlike the  conditional Shannon entropy, there is  no universally accepted definition for the conditional R\'enyi entropy, so we define the quantities that we study carefully in Section \ref{sec:info_measures}.   Extensive discussions of various plausible notions of the conditional R\'enyi entropy  are provided in the recent works by Teixeira, Matos and  Antunes \cite{teixeira2012} and Fehr and Berens~\cite{fehr}. 

We motivate our study by first showing that the limits of the (normalized)  remaining uncertainty  $\frac{1}{n} H_{1 + s} (A^n | f(A^n),E^n)$ and the exponent of the remaining uncertainty $-\frac{1}{n} \log H_{1 + s} (A^n | f(A^n),E^n)$ (for appropriately chosen R\'enyi parameters $1+s$) are, respectively,  {\em generalizations} of  the strong converse exponent and the error exponent for decoding $A^n$ given $(f(A^n), E^n)$. 
Recall that the strong converse exponent~\cite{Dueck79,Oohama94} is the exponential rate at which the probability of {\em correct} decoding tends to zero when one operates at a rate {\em below} the first-order coding rate, i.e., the conditional Shannon entropy $H(A|E)$. In contrast, the error exponent~\cite{gallagerIT,Gal76,Kos77,Csi97} is the exponential rate at which the probability of {\em incorrect} decoding tends to zero when one operates  at a rate  {\em above} $H(A|E)$. Thus, studying the asymptotics of the conditional R\'enyi entropy allows us not only to understanding the remaining uncertainty for various classes of {\em hash functions}~\cite{carter79, Wegman81}  but also allows us to provide additional information and intuition concerning  the strong converse exponent and the error exponent for Slepian-Wolf coding~\cite{sw73}. We also motivate our study by considering a scenario in information-theoretic security where the hash functions we study appear naturally, and coding can be done in a computationally efficient manner. The present work can be regarded  a follow-on from the authors' previous work  in~\cite{HayashiTan15} on the asymptotics  of the equivocations where we studied the behavior of $C_{1+s} :=nR-H_{1 + s} (f(A^n) |  E^n)$ and $C_{1+s}^\uparrow:= nR-H_{1 + s}^\uparrow (f(A^n) |  E^n)$ (where $R = \frac{1}{n}\log \| f\|$ is the rate of the  cardinality of the range of $f$). In \cite{HayashiTan15}, we also studied the  exponents and second-order asymptotics of the equivocation. However, we note that because we consider the remaining uncertainty instead of the equivocation, several novel techniques, including new one-shot bounds and large-deviation techniques, have to be developed to single-letterize various expressions.

\begin{figure} 
\centering
\setlength{\unitlength}{.4mm}
\begin{picture}(200, 60)
\thicklines
\put(-15, 45){\vector(1, 0){30}}
 
\put(45, 45){\vector(1, 0){60}}
\put(135, 45){\vector(1, 0){30}}

\put(120, 0){\vector(0,1){30}}

\put(15, 30){\line(1, 0){30}}
\put(15, 30){\line(0,1){30}}
\put(45, 30){\line(0,1){30}}
\put(15, 60){\line(1,0){30}}

\put(105, 30){\line(1, 0){30}}
\put(105, 30){\line(0,1){30}}
\put(135, 30){\line(0,1){30}}
\put(105, 60){\line(1,0){30}}

\put(-15, 49){  $A^n$}
\put(60, 50){  $f_n(A^n)$}
 \put(103, 0){ $E^n$}
\put(153, 49){  $\hatA^n$} 
\put(140, 33){  $H_{1\pm s}(A^n|f_n (A^n),E^n)$} 
\put(140, 17){  $H_{1\pm s}^\uparrow(A^n|f_n (A^n),E^n)$} 
\put(27, 42){$f_n$ } 
\put(117, 42){$g_n$} 
%
  \end{picture}
  \caption{The Slepian-Wolf~\cite{sw73}  source coding problem. We are interesting in quantifying the asymptotic behaviors of  the remaining uncertainty of $A^n$ given $( f_n(A^n), E^n)$ measured according to the conditional R\'enyi entropies  $H_{1\pm s}$ and $H_{1\pm s}^\uparrow$ defined in  \eqref{eqn:renyi_ent_Q}  and \eqref{eqn:gallager_form}.  }
  \label{fig:sw}
\end{figure}
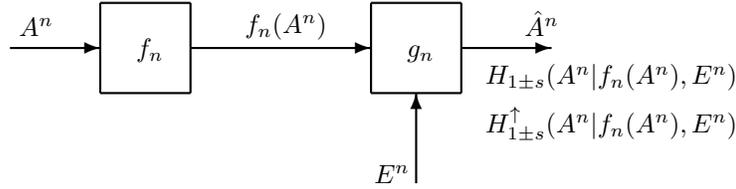

\subsubsection*{Paper Organization}

This paper is organized as follows. In Section~\ref{sec:prelims}, we recap the definitions of standard Shannon  information measures and some less common R\'enyi  information measures~\cite{fehr,teixeira2012}. We also introduce some new quantities and state relevant properties of the information measures. We state some notation concerning the method of types~\cite{Csi97}. In Section~$  $\ref{sec:motivate}, we further motivate our study by relating the quantities we wish to characterize to the error exponent and strong converse exponent of Slepian-Wolf coding (Proposition~\ref{pr:err_exp}). In Section \ref{sec:equiv}, we define various important classes of hash functions~\cite{carter79, Wegman81} (such as universal$_2$  and strong universal hash functions) and further motivate the study of the quantities of interest by  discussing    efficient implementations of universal$_2$ hash functions via circulant matrices~\cite{HayashiT2013}. The final parts of Section \ref{sec:equiv} contain  our main results concerning the asymptotics of the normalized remaining uncertainties (Theorem \ref{thm:rem}),  the optimal rates of compression of the main source to ensure that the remaining uncertainties vanish (Theorem~\ref{cor:threshold}), and the exponents of the remaining uncertainties (Theorem~\ref{thm:exponents}). We show that the optimal rates are tight in certain  ranges of the R\'enyi parameter.  For these evaluations, we make use of several novel one-shot bounds, large-deviation techniques as well as the moments of type class enumerator method~\cite{merhav08,merhav14, merhav_FnT,kaspi11,merhav13}.   Theorems \ref{thm:rem}, \ref{cor:threshold} and~\ref{thm:exponents} are proved in Sections \ref{sec:prf_rem}, \ref{sec:threshold_prf} and~\ref{sec:prf_exp} respectively.  We conclude our discussion and suggests further avenues for research in Section~\ref{sec:con}. Some technical results (e.g.,   one-shot bounds, concentration inequalities) are relegated to the appendices.

\section{Information Measures and Other Preliminaries }\label{sec:prelims}
\subsection{Basic Shannon and R\'enyi Information Quantities} \label{sec:info_measures}
We now introduce some information measures that generalize Shannon's information measures. Fix a normalized distribution $P_A \in\calP(\calA)$ and a sub-distribution (a non-negative vector but not necessarily summing to one) $Q_A\in\bar{\calP}(\calA)$ supported on a finite set $\calA$. Then the {\em relative entropy} and the {\em R\'enyi divergence of order $1+s$} are respectively defined as 
\begin{align}
D(P_A  \| Q_A)  &:= \sum_{a\in\calA}P_A(a) \log\frac{P_A(a)}{Q_A(a)} \\*
D_{1+s}(P_A  \| Q_A)  &:=  \frac{1}{s}\log\sum_{a\in\calA} P_A(a)^{1+s}Q_A(a)^{-s},
\end{align}
where throughout, $\log $ is to the natural base $\rme$.  
It is known that $\lim_{s\to 0} D_{1+s}(P_A\| Q_A) = D(P_A\| Q_A)$ so a special (limiting) case of the  R\'enyi divergence is the usual relative entropy. It is  also known that  the map $s\mapsto s D_{1+s}(P_A\| Q_A)$ is concave in $s \in\bbR$ and hence $D_{1+s}(P_A\| Q_A)$ is monotonically increasing for $s\in\bbR$. Furthermore, the following {\em data processing} or {\em information processing inequalities} for R\'enyi divergences hold for $s\in [-1,1]$,
\begin{align}
D(P_A W  \| Q_A W) &\le D(P_A  \| Q_A)\\
D_{1+s}(P_A W  \| Q_A W) &\le D_{1+s}(P_A  \| Q_A) . \label{eqn:dpi_rd}
\end{align}
Here $W:\calA\to\calB$ is any stochastic matrix (channel) and $P_AW(b):=\sum_a W(b|a)P_A(a)$ is the output distribution induced by $W$ and $P_A$.

 We also introduce conditional entropies on the product alphabet  $\calA\times\calE$ based on the divergences above. Let $I_A(a)=1$ for each $a\in\calA$. 
If $P_{AE}$ is a  distribution  on $\calA\times\calE$, the {\em conditional entropy}, the {\em conditional R\'enyi entropy of order $1 + s$ }  and the {\em min-entropy} {\em relative to another normalized distribution $Q_E$ on $ \calE$}  as
\begin{align}
H(A|E|P_{AE}\| Q_E)  &:= - D(P_{AE} \| I_{A}\times Q_E) , \label{eqn:cond_entr_given} \\
H_{1+s}(A|E|P_{AE}\| Q_E)  &:= - D_{1+s}(P_{AE} \| I_{A}\times Q_E) ,\label{eqn:r_cond_entr_given} \\
H_{\min}(A|E|P_{AE}\| Q_E)  &= -\log\max_{(a,e):Q_E(e)>0} \frac{P_{AE}(a,e)}{Q_E(e)} .
\end{align}
It is known that $\lim_{s\to 0 } H_{1+s}(A|E|P_{AE}\| Q_E)=H(A|E|P_{AE}\| Q_E)$ and  
\begin{equation}
\lim_{s\to \infty} H_{1+s}(A|E|P_{AE}\| Q_E)=H_{\infty}(A|E|P_{AE}\| Q_E)=H_{\min}(A|E|P_{AE}\| Q_E) .\label{eqn:min_ent_lim}
\end{equation}
If $Q_E=P_E$, we simplify the above notations and denote the {\em conditional entropy}, the {\em conditional R\'enyi entropy of order $1 + s$} and the {\em min-entropy} as
\begin{align}
H(A|E|P_{AE}) &:= H(A|E|P_{AE}\|P_E) =-\sum_e P_E(e)\sum_a P_{A|E}(a|e) \log P_{A|E}(a|e),\\
H_{1+s}(A|E|P_{AE}) &:= H_{1+s}(A|E|P_{AE}\|P_E) = -\frac{1}{s}\log\sum_{e} P_E(e)\sum_a P_{A|E}(a|e)^{1+s}, \label{eqn:renyi_ent_Q} \\
H_{\min}(A|E|P_{AE}) &:= H_{\min}(A|E|P_{AE}\|P_E)= -\log\max_{(a,e):P_E(e)>0} P_{A|E}(a | e).
\end{align}
The map $s\mapsto s H_{1+s}(A|E|P_{AE})$ is concave, and  $H_{1+s}(A|E|P_{AE}\|Q_E)$ is 
monotonically decreasing  for  $s\in\bbR\setminus\{0\}$.  The definition of the conditional R\'enyi entropy in \eqref{eqn:renyi_ent_Q} is due to Hayashi~\cite[Section II.A]{Hayashi11} and \v{S}kori\'{c} {\em et al.}~\cite[Definition 7]{skoric11}. 

We are also interested in the  so-called {\em Gallager form} of the conditional R\'enyi entropy  and the min-entropy for a  joint distribution $P_{AE} \in\calP(\calA\times\calE)$:
\begin{align}
H_{1+s}^{\uparrow} (A|E|P_{AE})   &:=  -\frac{1 +  s}{s}\log\sum_e P_E(e)\bigg( \sum_a P_{A|E}(a|e)^{1+s}\bigg)^{\frac{1}{1+s}} \label{eqn:gallager_form} \\
H_{\min}^{\uparrow} (A|E|P_{AE}) &    :=H_{\infty}^{\uparrow} (A|E|P_{AE}) =  -\log\sum_eP_E(e) \max_a P_{A|E}(a|e) \label{eqn:min_gal} 
\end{align}
By defining the familiar  {\em Gallager function}  \cite{gallagerIT,Gal76}  (parametrized slightly differently)
\begin{equation}
\phi\big (s|A|E|P_{AE} \big):=\log\sum_e P_E(e)\bigg( \sum_a P_{A|E}(a|e)^{ \frac{ 1}{1-s}}\bigg)^{1-s}
\end{equation}
we can express \eqref{eqn:gallager_form} as 
\begin{equation}
H_{1+s}^{\uparrow} (A|E|P_{AE})=-\frac{1+s}{s}\phi\bigg( \frac{s}{1+s}\Big|A |E |P_{AE}\bigg),
\end{equation}
thus (loosely) justifying the nomenclature ``Gallager form'' of the conditional R\'enyi entropy in \eqref{eqn:gallager_form}. Note that $H_{1+s}$ and  $H_{1+s}^\uparrow$ are respectively denoted as $\tilH_{1+s}^4$ and $H_{1+s}$ in the paper by Fehr and Berens~\cite{fehr}. The Gallager form of the  conditional R\'enyi entropy, also commonly known as {\em Arimoto's conditional R\'enyi entropy}~\cite{arimoto75}, was shown in~\cite{fehr} to satisfy two natural  properties for $s\ge -1$, namely,  {\em monotonicity under conditioning} (or simply {\em monotonicity})
\begin{equation}
H_{1+s}^\uparrow (A|B,E|P_{ABE})\le H_{1+s}^\uparrow(A|E|P_{AE}) , \label{eqn:mono}
\end{equation}
 and the {\em chain rule}
\begin{equation}
H_{1+s}^\uparrow (A|B,E|P_{ABE})\ge H_{1+s}^\uparrow(A|E|P_{AE})-\log|\calB|.\label{eqn:cr}
\end{equation}
  The monotonicity  property  of $H_{1+s}^\uparrow$ was also shown operationally by Bunte and Lapidoth in the context of lossless source coding with lists and side-information~\cite{bunte13} and encoding tasks with side-information~\cite{bunte14}.   We exploit these properties in the sequel. The quantities $H_{1+s}$ and $H_{1+s}^{\uparrow}$ can be   shown to be related as follows~\cite[Theorem 4]{fehr}
\begin{align}
\max_{Q_E \in \calP(\calE)} H_{1+s}(A|E|P_{AE}\| Q_E) = H_{1+s}^{\uparrow}(A|E|P_{AE}) \label{eqn:ent_min}
\end{align}
for $s\in [-1,\infty)\setminus\{0\}$. The maximum on the left-hand-side is attained for the tilted distribution
\begin{equation}
Q_E(e) = \frac{ (\sum_a P_{AE}(a,e)^{1+s})^{\frac{1}{1+s}} }{\sum_e(\sum_a P_{AE}(a,e)^{1+s})^{\frac{1}{1+s}}}. \label{eqn:Q_tilt}
\end{equation}
The map $s\to s H_{1+s}^{\uparrow}(A|E|P_{AE})$  is concave and the map $s\mapsto H_{1+s}^{\uparrow}(A|E|P_{AE})$ is monotonically decreasing for $s\in (-1,\infty)$. It can be shown by L'H\^{o}pital's rule that $
\lim_{s\to 0}  H_{1+s}^\uparrow(A|E|P_{AE}) = H(A|E|P_{AE})$. 
Thus, we regard $ H_{1}^\uparrow(A|E|P_{AE} )$ as $H(A|E|P_{AE})$, i.e., when $s=0$, the conditional R\'enyi entropy and its Gallager form coincide and are equal to the  conditional Shannon entropy.  

 We also find it useful to consider a  {\em two-parameter family  of the conditional  R\'enyi entropy}\footnote{This {\em new} information-theoretic quantity is  somewhat related to $H_{1+\theta, 1+\theta'}$ in the work by  Hayashi and Watanabe~\cite[Eq.~(14)-(15)]{HayashiW2013} but is different and not to be confused with $H_{1+\theta, 1+\theta'}$. }
 \begin{align}
H_{1+s| 1+t}(A|E|P_{AE}) := -\frac{1 + t}{s}\log\sum_{ e} P_E(e) \bigg(\sum_aP_{A|E}(a|e)^{1+s}  \bigg)
 \bigg(\!\sum_{\tila}P_{A|E}(\tila|e)^{1+t} \bigg)^{-\frac{s}{1+t}}. \label{eqn:two_param} 
 \end{align}
 Clearly $ H_{1+s| 1+s }(A|E|P_{AE})=H_{1+s}^{\uparrow} (A|E|P_{AE})$,  so the two-parameter conditional  R\'enyi entropy is  a generalization of the Gallager form  of the conditional R\'enyi entropy in~\eqref{eqn:gallager_form}.  


For future reference, given  a joint source $P_{AE}$, define the {\em critical rates}
\begin{align}
\hatR_s &:=\frac{\rmd }{\rmd t} \, tH_{1+t}(A|E|P_{AE})\Big|_{t=s } , \quad\mbox{and}\label{eqn:crit_rate1} ,\\
\hatR_s^\uparrow &:=\frac{\rmd}{\rmd t} tH_{1+t|1+s} (A|E|P_{AE} )  \Big|_{t=s} . \label{eqn:crit_rate2}
\end{align}

\subsection{Notation for Types} \label{sec:types}
The proofs  of our results leverage on the {\em method of types}~\cite[Ch.~2]{Csi97}, so we summarize some relevant notation here. The set of all distributions (probability mass functions) on a finite set $\calA$ is denoted as $\calP(\calA)$. The {\em type} or {\em empirical distribution} of a sequence $\ba\in\calA^n$ is the distribution $Q(a) = \frac{1}{n}\sum_{i=1}^n\bone\{ a_i = a\},a\in\calA$.  The set of all sequences $\ba\in\calA^n$ with  type  $Q \in\calP(\calA)$ is the {\em type class} and is denoted as $\calT_Q \subset\calA^n$. The set of all $n$-types (types formed from length-$n$ sequences) on alphabet $\calA$ is denoted as $\calP_n(\calA)$. When we write $a_n\dotleq b_n$, we mean that inequality on an exponential scale, i.e., $\varlimsup_{n\to\infty}\frac{1}{n}\log\frac{a_n}{b_n}\le 0$. The notations $\dotgeq$ and $\doteq$ are defined analogously. Throughout, we will use the fact that the number of types $|\calP_n(\calA)|\le (n+1)^{|\calA|}\doteq 1$.

\section{Motivation for Studying Remaining Uncertainties} \label{sec:motivate}
As mentioned in the introduction, in this paper, we study the remaining uncertainty and its rate of exponential decay measured using various R\'enyi information measures. In this section, we further motivate the relevance of this study by relating the remaining uncertainty to the strong converse exponent  for decoding $\ba = (a_1,a_2,\ldots, a_n)\in\calA^n$ given side information $\be = (e_1,e_2,\ldots, e_n)\in\calE^n$ and the compressed version of $\ba$, namely $m=f(\ba)$ (Slepian-Wolf problem). We also relate the exponential rate of decay of the remaining uncertainty for a source coding rate above the first-order fundamental limit  to the error exponent of the Slepian-Wolf problem.

\subsection{Relation to the Strong Converse Exponent for Slepian-Wolf Coding }
Consider the Slepian-Wolf source coding problem as shown in Fig.~\ref{fig:sw}. For a given   function (encoder) $f_n:\calA^n\to\calM_n$ and side information vector $\be\in\calE^n$, we may define the {\em maximum a-posteriori (MAP)} decoder $g_{f_n}: \calM_n\times\calE^n\to\calA^n$ as follows:
\begin{align}
g_{f_n}(m,\be)& :=\argmax_{ \ba \in\calA^n : f_n(\ba) = m } P_{AE}^n(\ba,\be)   
  =\argmax_{ \ba \in\calA^n : f_n(\ba) = m } P_{A|E}^n(\ba | \be)  \label{eqn:def_gf}  . 
\end{align}
Define the probability of correctly decoding $\ba$ given the encoder $f_n$ and the MAP decoder $g_{f_n}$ as follows:
\begin{align}
\rmP_{\rmc}^{(n)}(f_n) := \sum_{\be} \sum_{\ba : \ba = g_{f_n}(f_n(\ba),\be)}  P_{AE}^n(\ba,\be)
\end{align}
Then, by the definition of $H_{\infty}^\uparrow$ in \eqref{eqn:min_gal}, we immediately see that 
\begin{equation}
-\frac{1}{n}\log \rmP_{\rmc}^{(n)}(f_n )=\frac{1}{n}H_{\infty}^\uparrow(A^n | f_n(A^n),E^n  | P_{AE}^n ). \label{eqn:str_conv_exp}
\end{equation}
When optimized over $\{ f_n \}_{n=1}^\infty$, the quantity on the left of \eqref{eqn:str_conv_exp} (or its limit) is called the {\em strong converse exponent} as it characterizes the optimal  exponential rate at which the  probability of {\em correct} decoding   the true source $\ba$  given $(f_n(\ba),\be)$ decays to zero.   Thus, by studying the asymptotics of $\frac{1}{n}H_{1+s}^\uparrow$ for all $s\in [0,\infty)$ and, in particular, the limiting case   of $s\uparrow \infty$ (which we do in \eqref{eqn:plus_max_uni_str}  in Part (2) of Theorem~\ref{thm:rem}), we obtain a generalization of  the strong converse exponent for the Slepian-Wolf problem. In fact, it is known that $\varliminf_{n\to\infty}-\frac{1}{n}\log \rmP_{\rmc}^{(n)}(f_n)>0$ for any sequence of encoders $\{f_n\}_{n=1}^\infty$ if and only if the rate $\varlimsup_{n\to\infty}\frac{1 }{n}\log \|f_n\|< H(A|E|P_{AE})$ \cite[Theorem 2]{Oohama94}. This  fact will be utilized in  the proof of Theorem~\ref{cor:threshold}.
\subsection{Relation to the Error Exponent for Slepian-Wolf Coding }
Similarly, we may define the probability of {\em incorrectly} decoding $\ba$ given the encoder $f_n$ and MAP decoder $g_{f_n }$ as follows:
\begin{equation}
\rmP_{\rme}^{(n)}(f_n) := \sum_{\be} \sum_{\ba : \ba \ne g_{f_n}(f_n(\ba),\be)}  P_{AE}^n(\ba,\be). \label{eqn:perr}
\end{equation}
Then we have the following proposition concerning the exponent of $\rmP_{\rme}^{(n)}(f_n)$.
\begin{proposition} \label{pr:err_exp}
Assume that $\rmP_{\rme}^{(n)}(f_n) $ tends to zero exponentially fast for a given sequence of hash functions $ \{f_n \}_{n=1}^{\infty}$, i.e., $\lim_{n\to\infty}-\frac{1}{n}\log \rmP_{\rme}^{(n)}(f_n) >0$ (the existence of the limit is part of the assumption). Then for any $s\ge 1$, we have 
\begin{align}
\lim_{n\to\infty} -\frac{ 1}{n}\log \rmP_{\rme}^{(n)}(f_n) &  = \lim_{n\to\infty}-\frac{1}{n}\log H_{1+s}(A^n | f_n(A^n),E^n | P_{AE}^n) \label{eqn:err_exp1} \\
&  = \lim_{n\to\infty}-\frac{1}{n}\log H_{1+s}^\uparrow(A^n | f_n(A^n),E^n | P_{AE}^n) .\label{eqn:err_exp2}
\end{align}
\end{proposition}
We recall, by the Slepian-Wolf theorem~\cite{sw73}, that there exists a sequence of encoders $\{f_n\}_{n=1}^\infty$ such that  $\rmP_{\rme}^{(n)}(f_n)$ tends to zero  if and only if   $\varliminf_{n\to\infty}\frac{1 }{n}\log \|f_n\|\ge H(A|E|P_{AE})$. 
When optimized over $\{ f_n \}_{n=1}^\infty$,  the quantity on the left of \eqref{eqn:err_exp1}  is called the {\em optimal error exponent}   and it characterizes the optimal exponential rate at which the error probability of decoding $\ba$ given $(f_n(\ba),\be)$ decays to zero. Thus, Proposition~\ref{pr:err_exp} says that the exponents of $H_{1  + s}$ and $H_{1+ s}^\uparrow$   for   $s\ge 1$ are generalizations of the error exponent of decoding $A^n$ given $(f_n(A^n),E^n)$. We establish bounds on these limits for certain classes of hash functions in Part~(2) of Theorem \ref{thm:exponents}.
\begin{proof}
We first consider the Gallager form of the conditional R\'enyi entropy $H_{1+s}^\uparrow$. For brevity, we let  $f=f_n$  (suppressing the dependence on $n$) and we also define the probability distributions $P: =P_{f(A^n) , E^n}$ and  $Q:= P_{A^n| f(A^n),E^n}$. Recall the definition of the MAP decoder $g_f(m,\be) $ in \eqref{eqn:def_gf}. We have
\begin{align}
\rme^{-s H_{1+s}^\uparrow(A^n| f(A^n), E^n | P_{AE}^n)  }  &=\sum_{\be,m} P(m,\be) \left(\sum_{\ba} Q(\ba | m,\be)^{1+s} \right)^{\frac{1}{1+s}}\\
&\ge\sum_{\be,m} P(m,\be) \left(  Q(g_f( m,\be) | m,\be)^{1+s} \right)^{\frac{1}{1+s}}\\
&= \sum_{\be,m} P(m,\be)  Q(g_f( m,\be) | m,\be )\\
&=1-\rmP_{\rme}^{(n)}(f) . \label{eqn:ee_gal}
\end{align}
In the following chain of inequalities, we will employ    Taylor's theorem with the Lagrange form of the remainder for the function $t\mapsto  (1+t)^{1+s}$  at $t=0^-$, i.e.,
\begin{equation}
 (1+t)^{1+s} = 1 + (1+s) t + \frac{ s (1+ s )}{2}(1+\xi)^{s-1} t^2\label{eqn:taylor0} 
\end{equation}
for some $\xi\in [t,0]$.  We choose $t$ to be  $Q(g_f( m,\be) | m,\be )-1$ in  our application in~\eqref{eqn:taylor1} to follow.  Let $\xi(m,\be)$ be a generic element of $[Q(g_f( m,\be) | m,\be )-1,0]\subset [-1,0]$ taking the role of $\xi$ in the Taylor series expansion in~\eqref{eqn:taylor0}. We bound the   conditional R\'enyi entropy as follows:
\begin{align}
&\rme^{-s H_{1+s} (A^n| f(A^n), E^n | P_{AE}^n)  } \nn\\*
&=\sum_{\be,m} P(m,\be) \sum_{\ba}  Q(\ba | m,\be)^{1+s} \\
&=\sum_{\be,m} P(m,\be) \Bigg\{ Q(g_f(m,\be) | m,\be)^{1+s}  + \sum_{\ba : \ba\ne g_f(m,\be)}  Q(\ba | m,\be)^{1+s} \Bigg\} \\
& = \sum_{\be,m} P(m,\be)  \Bigg\{  1+(1+s) \big[ Q(g_f( m,\be) | m,\be )-1\big]  \nn\\*
&\qquad + \frac{s (1+s)}{2} \big(1+\xi(m,\be)\big)^{s-1} \big[ Q(g_f( m,\be) | m,\be )-1\big]^2 +   \sum_{\ba : \ba\ne g_f(m,\be)}  Q(\ba | m,\be)^{1+s}  \Bigg\}\label{eqn:taylor1} \\
& \le  \sum_{\be,m} P(m,\be)  \Bigg\{  1+(1+s) \big[ Q(g_f( m,\be) | m,\be )-1\big]  \nn\\*
&\qquad + \frac{s (1+s)}{2} \big[ Q(g_f( m,\be) | m,\be )-1\big]^2 +  \sum_{\ba : \ba\ne g_f(m,\be)}  Q(\ba | m,\be)  \Bigg\} \label{eqn:taylor2}  \\
&  =  1+\sum_{\be,m} P(m,\be) \Bigg\{  ( - (1+s) +1 )\sum_{\ba : \ba\ne g_f(m,\be)}  Q(\ba | m,\be)  +   \frac{s (1+s)}{2}  \bigg[ \sum_{\ba : \ba\ne g_f(m,\be)}  Q(\ba | m,\be) \bigg]^2     \Bigg\}\\
&= 1-s \sum_{\be,m} P(m,\be) \sum_{\ba : \ba\ne g_f(m,\be)}  Q(\ba | m,\be)    + \frac{s (1+s)}{2}  \sum_{\be,m} P(m,\be)  \bigg[  \sum_{\ba : \ba\ne g_f(m,\be)}  Q(\ba | m,\be)\bigg]^2\label{eqn:ee_no_gal0} \\
&=1-s \rmP_{\rme}^{(n)}(f)  + \frac{s (1+s)}{2}  \sum_{\be,m} P(m,\be)  \bigg[  \sum_{\ba : \ba\ne g_f(m,\be)}  Q(\ba | m,\be)\bigg]^2 .\label{eqn:ee_no_gal}
\end{align}
In \eqref{eqn:taylor2}, noting that $s - 1\ge 0$, we uniformly upper bounded $\big(1+\xi(m,\be) \big)^{s-1}$ by $1$.  We also upper bounded $Q(\ba | m,\be)^{1+s}$ by $Q(\ba | m,\be)$.  In \eqref{eqn:ee_no_gal}, we used the definition of $\rmP_{\rme}^{(n)}(f)$ stated in \eqref{eqn:perr}.
Because $\rmP_{\rme}^{(n)}(f)$ is assumed to decay exponentially fast, we have 
\begin{align}
\rmP_{\rme}^{(n)}(f)  
&\doteq -\log\big(1- \rmP_{\rme}^{(n)}(f) \big) \label{eqn:log_ineq}\\*
&\ge sH_{1+s}^\uparrow (A^n| f(A^n), E^n | P_{AE}^n) \label{eqn:useee_gal}\\
&\ge sH_{1+s} (A^n| f(A^n), E^n | P_{AE}^n)  \label{eqn:ineq_H}\\
&\ge -\log\Bigg\{ 1-s \rmP_{\rme}^{(n)}(f)  +\frac{s (1+s)}{2}   \sum_{\be,m} P(m,\be)  \bigg[  \sum_{\ba : \ba\ne g_f(m,\be)}  Q(\ba | m,\be)\bigg]^2\Bigg\}\label{eqn:use_ee_no_gal}\\*
&\doteq s \rmP_{\rme}^{(n)}(f) - \frac{s (1+s)}{2}  \sum_{\be,m} P(m,\be)  \bigg[  \sum_{\ba : \ba\ne g_f(m,\be)}  Q(\ba | m,\be)\bigg]^2,\label{eqn:final_ee}
\end{align}
where \eqref{eqn:log_ineq} and \eqref{eqn:final_ee} follow  from $-\log(1-t)= t+O(t^2)$, \eqref{eqn:useee_gal} uses \eqref{eqn:ee_gal}, \eqref{eqn:ineq_H} uses the fact that $H_{1+s}^\uparrow\ge H_{1+s}$ (cf.\ \eqref{eqn:ent_min}) and~\eqref{eqn:use_ee_no_gal} uses~\eqref{eqn:ee_no_gal}. The second term in \eqref{eqn:final_ee}  is    exponentially smaller  than $ \rmP_{\rme}^{(n)}(f) $ because of the square  operation and  the fact that $\sum_{\ba : \ba\ne g_f(m,\be)}  Q(\ba | m,\be) <1$. Now, since $s\ge 1$ is constant, the exponents of the quantities on the left and right sides of the above chain are equal. Thus they are equal to the exponents of  $H_{1+s}^\uparrow (A^n| f(A^n), E^n | P_{AE}^n)$ and $H_{1+s} (A^n| f(A^n), E^n | P_{AE}^n)$ for every $s\ge 1$.  This completes the proof of Proposition~\ref{pr:err_exp}.
\end{proof}

\section{Main Results: Asymptotics of the Remaining Uncertainties} \label{sec:equiv}
In this section we present our   results concerning the asymptotic behavior of the remaining uncertainties and its exponential behavior. As mentioned in Section \ref{sec:motivate}, the former is a generalization of the strong converse exponent for the Slepian-Wolf problem~\cite{sw73}, while the latter is a generalization of the error exponent for the same problem. Before doing so, we define various classes of random hash functions and further motivate our analysis using an example from information-theoretic security.

\subsection{Definitions of Various Classes of  Hash Functions}
We now define  various classes of  hash functions.  We start by stating a slight generalization of the canonical definition  of a universal$_2$ hash function by Carter and Wegman~\cite{carter79}.

\begin{definition} \label{def:has}
 A {\em  random\footnote{For brevity, we will sometimes omit the qualifier ``random''. It is understood, henceforth, that all so-mentioned hash functions are random hash functions.}   hash function} $f_X$ is a stochastic map from $\calA$ to $\calM:= \{1, \ldots ,M\}$, where $X$ denotes a random variable describing its stochastic behavior. The set of all random hash functions mapping from $\calA$ to $\calM$ is denoted as $\calR = \calR(\calA,\calM)$. A hash 
function  $f_X$ is called an {\em $\epsilon$-almost  universal$_2$ hash function}  if it satisfies the following
condition: For any {\em distinct} $a_1, a_2\in\calA$, 
\begin{equation}
\Pr\big(f_X(a_1)=f_X(a_2) \big) \le \frac{\epsilon}{M}.\label{eqn:hash}
\end{equation}
   When $\epsilon=1$ in \eqref{eqn:hash}, we simply say that $f_X$ is a  {\em  universal$_2$ hash function~\cite{carter79}}. We denote the set of universal$_2$ hash functions  mapping from $\calA$ to $\calM$ by $\calU_2 =\calU_2(\calA,\calM)$. 
   \end{definition}
   The following definition is due to Wegman and Carter~\cite{Wegman81}.
   \begin{definition} \label{def:strong_has}
A random hash function $f_X : \calA\to\{1,\ldots, M\}$ is called {\em strongly universal} when the random variables $\{f_X(a) :  a\in\calA\}$ are independent and subject to a uniform distribution, i.e.,
\begin{align}
\Pr\big(f_X(a) = m \big) = \frac{1}{M} \label{eqn:st}
\end{align}
for all $m \in \{1,\ldots,M\}$.  If $f_X$ is a strongly universal hash function, we emphasize this fact by writing $\barf_{X}$. 
\end{definition}

\begin{figure}
\centering
\includegraphics[width = .45\columnwidth]{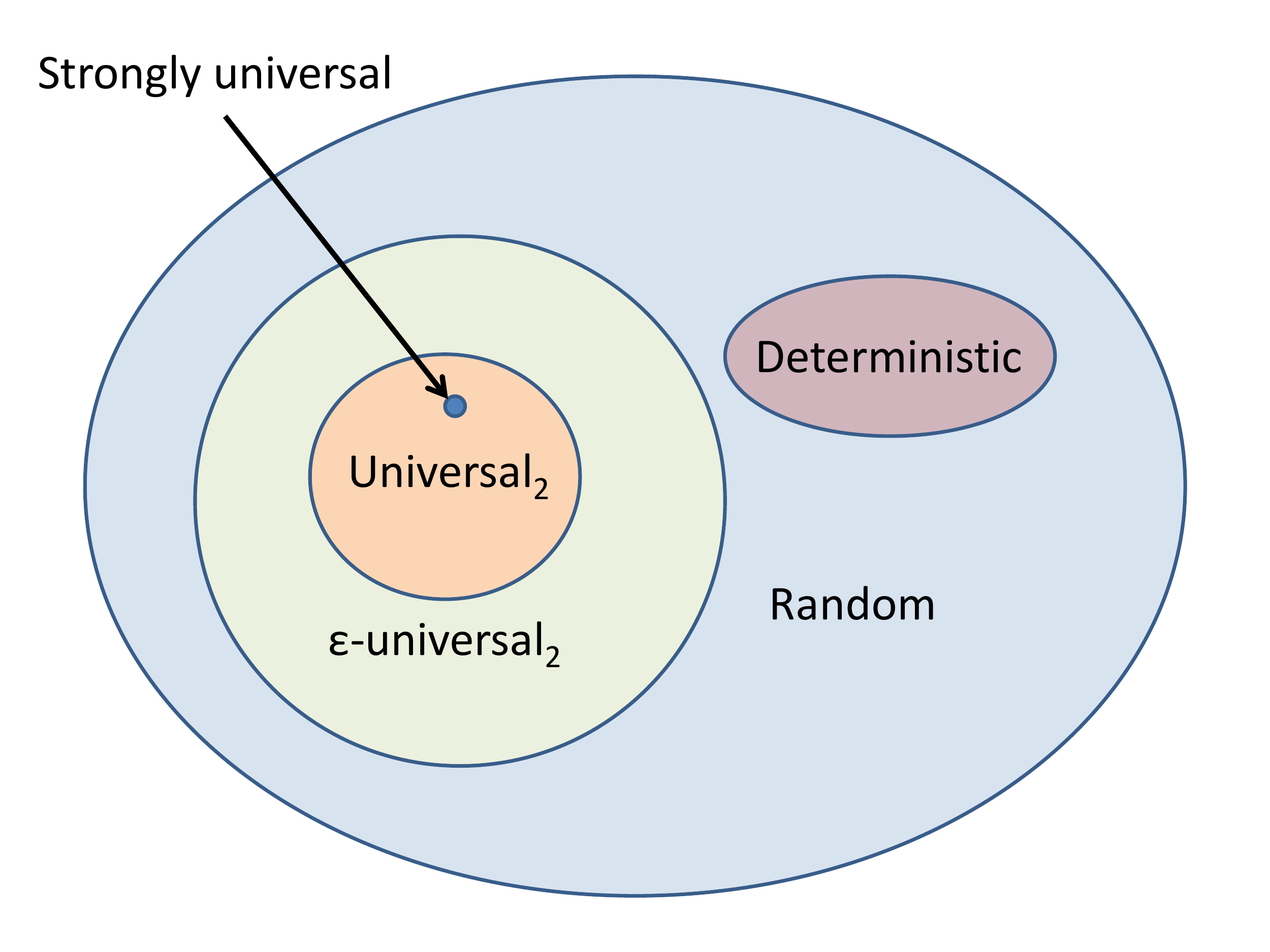}
\vspace{-.1in}
\caption{Hierarchy of hash functions. See Definitions  \ref{def:has} and \ref{def:strong_has}. }\vspace{-.1in}
\label{fig:venn}
\end{figure}

As an example, if $f_X$ independently and uniformly assigns each element of $a\in\calA$ into one of $M$ ``bins'' indexed by $m\in\calM$ (i.e., the familiar random binning process introduced by Cover in the context of Slepian-Wolf coding~\cite{cover75}), then  \eqref{eqn:st} holds,    
yielding a strongly universal  hash function.
The hierarchy of hash functions is shown in Fig.~\ref{fig:venn}. 


A universal$_2$ hash function $f$ can be  implemented efficiently via circulant (special case of Toeplitz) matrices. The  complexity is low---applying $f$ to an $m$-bit string requires $O(m\log m)$ operations generally. For details, see the discussion in Hayashi and Tsurumaru~\cite{HayashiT2013} and the subsection to follow. So, it is natural to assume that the encoding functions $f$ we analyze in this paper are  universal$_2$ hash functions.

\subsection{Another Motivation for Analyzing  Remaining Uncertainties} \label{sec:mot}

\begin{figure}
\centering
\includegraphics[width = .5\columnwidth]{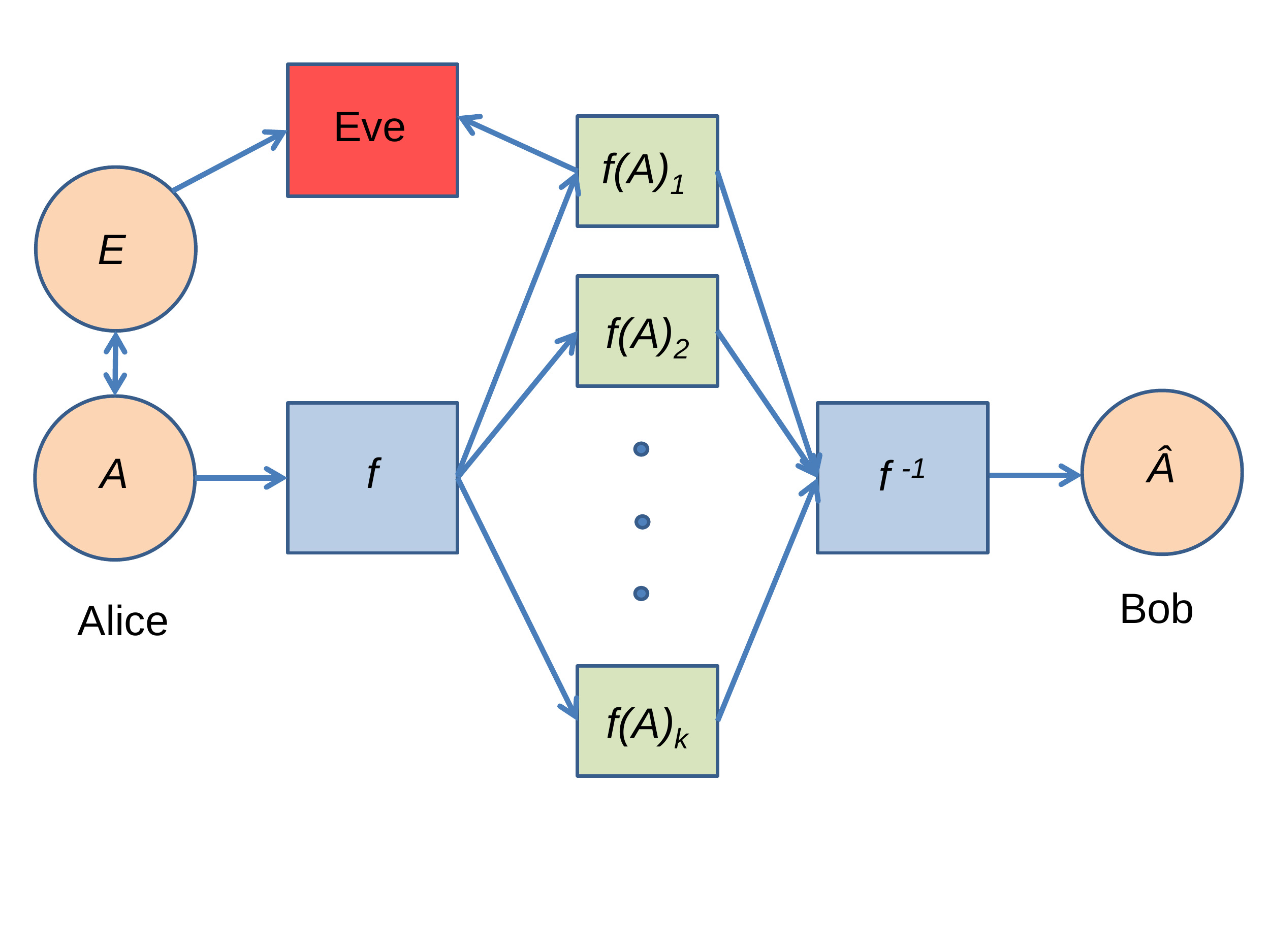}
\vspace{-.5in}
\caption{A secure communication scenario that motivates our study of remaining uncertainties. See Section \ref{sec:mot} for a discussion. } 
\label{fig:mot}
\end{figure}

To ensure a reasonable level of security in practice, we often send our message via multiple paths in networks. 
Assume that Alice wants to send an $m$-bit ``message'' $A$ to Bob via $l \in\bbN$ paths,
and that Eve has access to side-information $E$ correlated to $A$ and intercepts one of the $l$ paths.  
We also suppose   $m=kl$ for some $k\in\bbN$.
Alice applies an invertible function $f$ to $A$ and divides 
$f(A)$ into $k$ equal-sized parts 
$(f(A)_1, f(A)_2, \ldots, f(A)_k) \in \bbF_{2}^m 
\cong \bbF_2^l \oplus\ldots\oplus \bbF_2^l$ ($k$ times). See Fig.~\ref{fig:mot}.
Bob receives all of them, and applies $f^{-1}$ to decode $A$. Hence, Bob can  recover the original message $A$ losslessly.
However, if Eve somehow manages to tap on the $j$-th part $f(A)_j$ (where $j  \in \{1,2,\ldots, k\}$),  
Eve can possibly estimate the message $A$ from $E$ and $f(A)_j$ (in Fig.~\ref{fig:mot}, we assume Eve taps on the first piece of information $j=1$).
  Eve's  uncertainty with respect to $A$ is $H(A|f(A)_j, E|P_{A E})$ ($H$ here is a generic entropy function; it will be taken to be various conditional R\'enyi entropies in the subsequent subsections).
In this scenario, it is not easy to estimate the uncertainty $H(A|f(A)_j, E |P_{A E})$ as it depends on the choice of $j$.
To avoid such a difficulty,
we propose to apply a {\em random} invertible function $f_X$ to $A$.
To further resolve the aforementioned issue from a computational perspective, we regard $\bbF_{2}^m$ as the finite extension field $\bbF_{2^m}$.
When Alice and Bob choose invertible element $X$ in the finite field $\bbF_{2^m}$ subject to the uniform distribution,
and $f_X(A)$ is defined as $f(A):= XA$,
the map $A \mapsto f(A)_j$ is a universal$_2$ hash function.
Then, Eve's uncertainty with respect to $A$ can be described as  
$H(A|f(A)_j,E ,X|P_{AE}\times P_X)$.
When $(A,E)$ is taken to be $(A^n,E^n) = \{  (A_i,E_i) \}_{i=1}^n$ where the $(A_i,E_i)$'s are  independent and identically distributed, 
our results  in the following subsections are directly applicable  in evaluating Eve's uncertainty measured according to various conditional R\'enyi entropies. 
We remark that if $m$ is not a multiple of  $l$,
we can make the final block smaller than $l$ bits without any loss of generality asymptotically.

Indeed, this protocol can be efficiently implemented   with (low) complexity of $O(m\log m)$ \cite{HayashiT2013}  
because multiplication in the finite field $\bbF_{2^m}$
can be realized by an appropriately-designed  circulant matrix, leading to a fast Fourier transform-like algorithm.
Therefore, this communication setup, which contains an eavesdropper, is ``practical'' in the sense that encoding and decoding can be realized efficiently.


\subsection{Asymptotics of Remaining Uncertainties} \label{sec:asymp_ru}
Our results in Theorem \ref{thm:rem} to follow pertain to    the {\em worst-case} remaining uncertainties over  all universal$_2$ hash functions. We are interested in  $\sup_{f_{X_n} \in\calU_2 } \frac{1}{n}H_{1\pm s}$ and $\sup_{f_{X_n} \in\calU_2 }\frac{1}{n} H_{1\pm s}^\uparrow$, where $H_{1\pm s}$ is a shorthand  for $H_{1\pm s}( A^n | f_{X_n}(A^n) , E^n, X_n | P_{AE}^n\times P_{X_n} )$   (similarly for  $H_{1\pm s}^\uparrow$) and $P_{AE}^n$ is the $n$-fold product measure. We emphasize that the evaluations of $\sup_{f_{X_n} \in\calU_2 } \frac{1}{n}H_{1\pm s}$ and $\sup_{f_{X_n} \in\calU_2 }\frac{1}{n} H_{1\pm s}^\uparrow$ are {\em stronger} than those in standard achievability arguments in Shannon theory where one often uses a random selection argument to assert that an object (e.g.,  a code) with good properties exist. In our calculations of the asymptotics of  $\sup_{f_{X_n} \in\calU_2 } \frac{1}{n} H_{1\pm s}$ and $\sup_{f_{X_n} \in\calU_2 } \frac{1}{n}H_{1\pm s}^\uparrow$, we assert that {\em all}   hash functions in $\calU_2$ have a certain desirable property; namely, that the remaining uncertainties can be appropriately upper bounded. 
 In addition, in Theorem~\ref{cor:threshold} to follow, we also quantify  the minimum rate $R$ such that the {\em best-case} remaining uncertainties over {\em all random} hash functions  $\inf_{f_{X_n}\in\calR}\frac{1}{n} H_{1\pm s}$ and $\inf_{f_{X_n} \in\calR} \frac{1}{n} H_{1\pm s}^\uparrow$ vanish. For many values of $s$, we show the minimum rates for  the two different evaluations (worst-case over   all $f_{X_n}\in\calU_2$ and best-case over all $f_{X_n}\in\calR$) coincide, establishing tightness for the optimal compression rates.

Let $|t|^+ :=\max\{0,t\}$.  The following is our first main result.

\begin{theorem}[Remaining Uncertainties]\label{thm:rem}
For each $n\in\bbN$, let the size\footnote{When we write $M_n=\rme^{nR}$, we mean that $M_n$ is the integer $\lfloor\rme^{nR}\rfloor.$ } of the  range of $f_{X_n}$ be $M_n=\rme^{nR}$. Fix a joint distribution $P_{AE}\in\calP(\calA\times\calE)$.  Define the worst-case limiting  normalized remaining uncertainties   over all universal$_2$ hash functions as
\begin{align}
G(R,s) &:= \varlimsup_{n\to\infty}\frac{1}{n}\sup_{f_{X_n} \in\calU_2} H_{1 + s} ( A^n | f_{X_n}(A^n) , E^n, X_n | P_{AE}^n\times P_{X_n} ),\quad\mbox{and} \label{eqn:Gdef}\\
G^\uparrow(R,s)&:= \varlimsup_{n\to\infty}\frac{1}{n}\sup_{f_{X_n} \in\calU_2} H_{1 + s}^\uparrow ( A^n | f_{X_n}(A^n) , E^n, X_n | P_{AE}^n\times P_{X_n} ) \label{eqn:Gdef_up} .
\end{align}
Recall the definitions of the critical rates $\hatR_s$ and $\hatR_s^\uparrow$ in \eqref{eqn:crit_rate1} and \eqref{eqn:crit_rate2} respectively. 
The following achievability statements hold:
\begin{enumerate}
\item For any $s\in [0,1]$, we have 
\begin{equation}
G(R,-s)\le  |H_{1-s} (A|E |P_{AE}) - R|^+,\label{eqn:min_uni_res}
\end{equation}
 and for any $s \in [0,1/2]$, we have 
\begin{equation}
G^\uparrow(R,-s)\le  | H_{1-s}^\uparrow (A|E|P_{AE} )- R|^+.  \label{eqn:min_gal_res} 
\end{equation}
\item For $s\in (0,\infty)$, we have 
\begin{align}
G(R,s)\le \left\{ \begin{array}{cc}
H_{1+s}(A|E|P_{AE}) - R & R \le\hatR_s\\
\max_{ t\in [0,s]} \frac{t}{s} (H_{1+t}(A|E|P_{AE}) - R ) & R >\hatR_s
\end{array} \right.  , \label{eqn:plus_max_res}
\end{align}
and 
\begin{align}
G^\uparrow(R,s)\le \left\{ \begin{array}{cc}
  H_{1+s}^\uparrow(A|E|P_{AE}) - R   & R \le\hatR_s^\uparrow \\
\max_{ t\in [0,s]} \frac{t}{ s } (H_{1  +t|1+s}(A|E|P_{AE}) - R ) & R >\hatR_s^\uparrow .
\end{array} \right.  .\label{eqn:plus_max_uni_str} 
\end{align}
\end{enumerate}
\end{theorem}

Theorem \ref{thm:rem} is proved in Section \ref{sec:prf_rem} and uses several novel one-shot bounds on the remaining uncertainty (summarized in Appendix \ref{app:one-shot}) coupled with appropriate uses of   large-deviation results such as Cram\'er's theorem  and Sanov's theorem~\cite{Dembo}.

\begin{figure}[t]
\centering
\begin{tabular}{cc}
\includegraphics[width = .475\columnwidth]{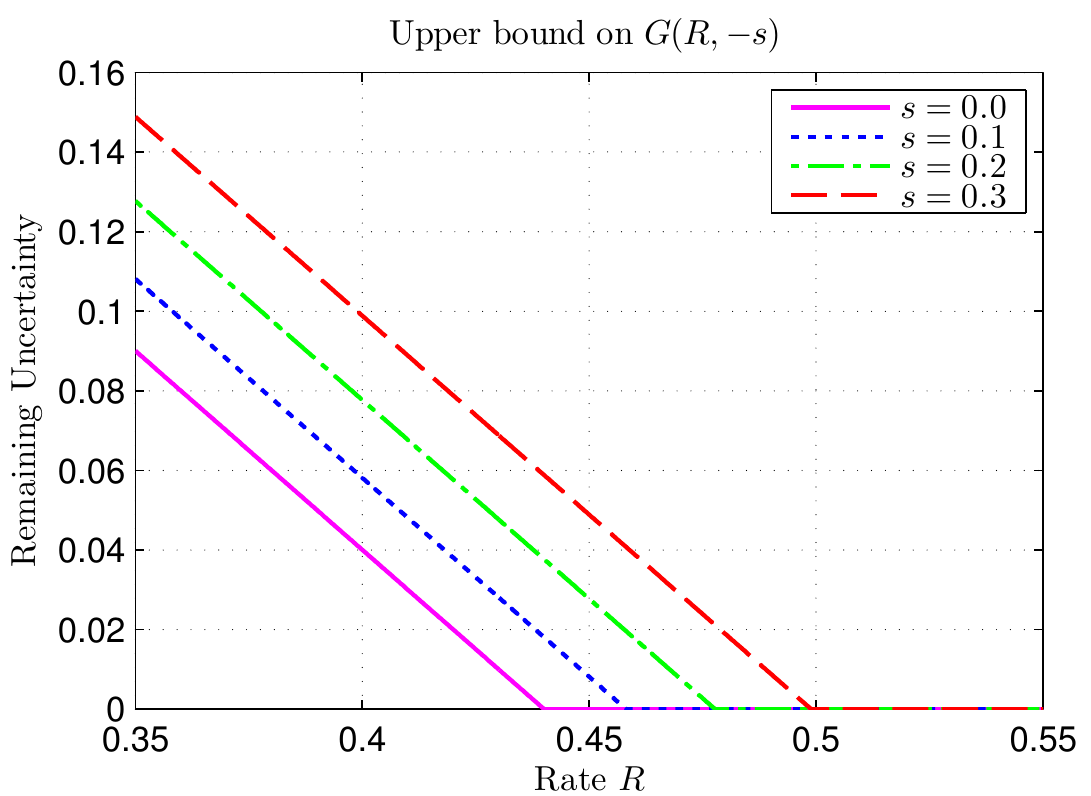} &
\includegraphics[width = .475\columnwidth]{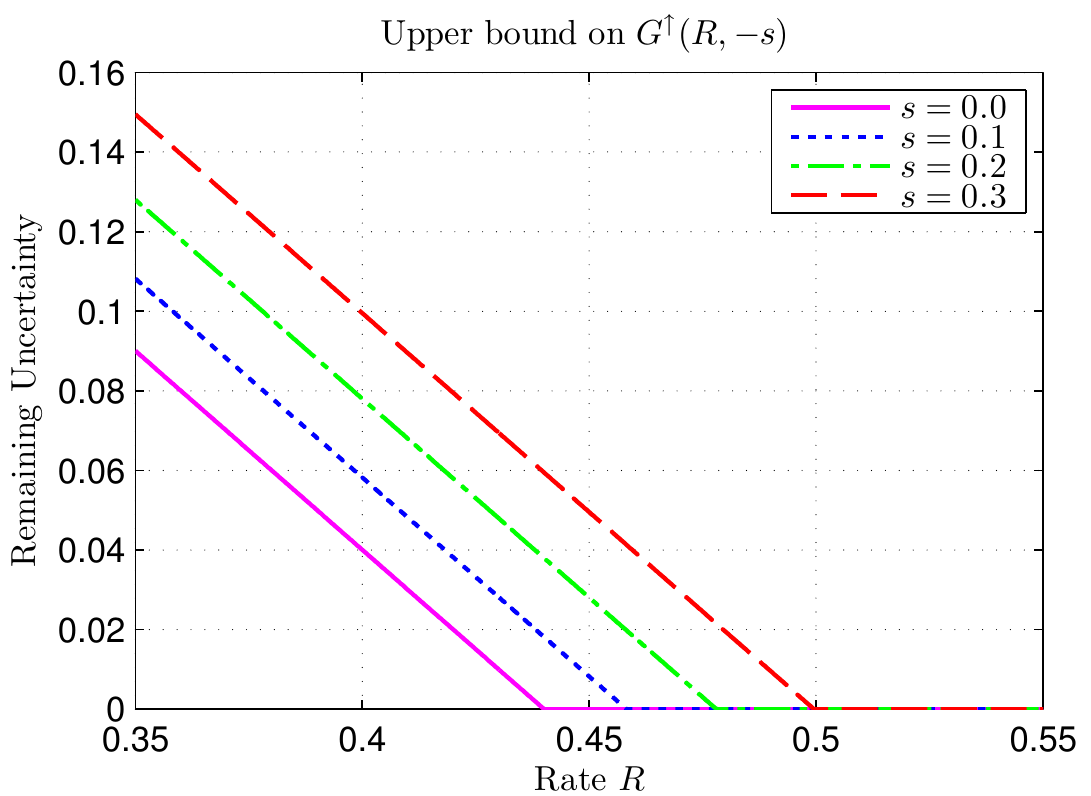} 
\end{tabular}
\caption{Illustration of the upper bounds on the remaining uncertainties $G (R,   -s) $ and $G^\uparrow (R,   -s) $ in \eqref{eqn:min_uni_res} and \eqref{eqn:min_gal_res} respectively. All curves transition from a positive quantity to zero at the R\'enyi conditional entropies $H_{1-s}(A|E|P_{AE})$ (left) and $H_{1-s}^{\uparrow}(A|E|P_{AE})$ (right).}
\label{fig:rm_simple}
\begin{tabular}{cc}
\includegraphics[width = .475\columnwidth]{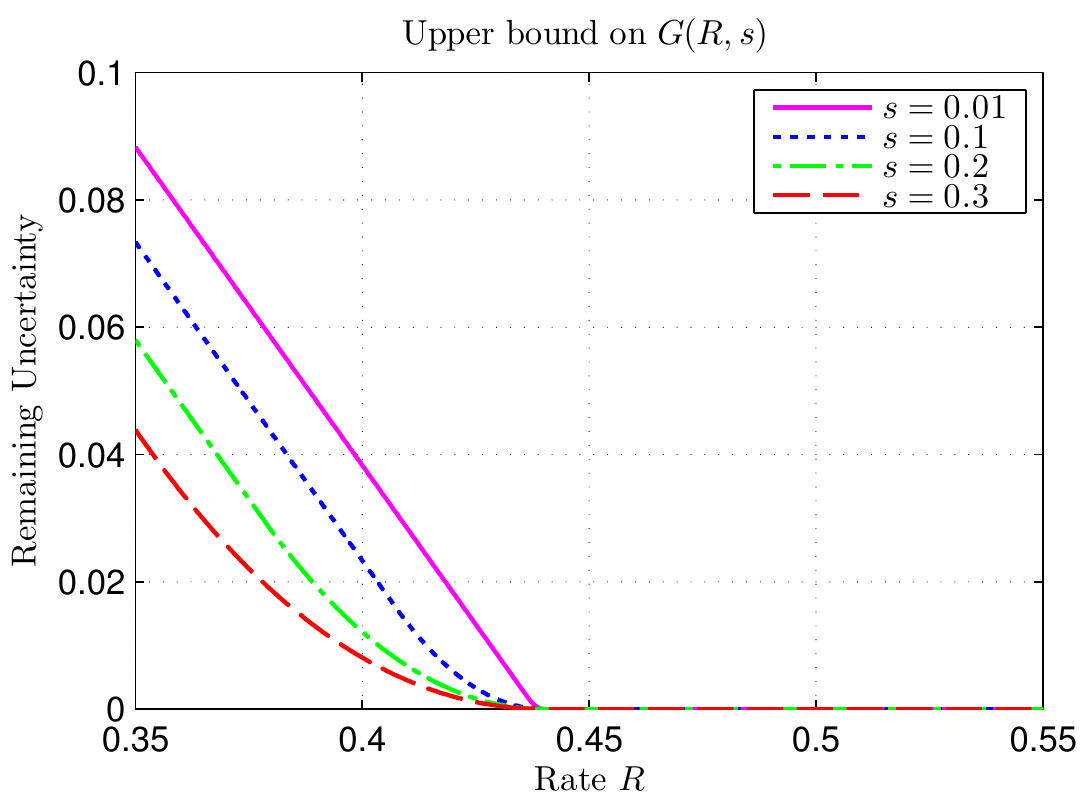} &
\includegraphics[width = .475\columnwidth]{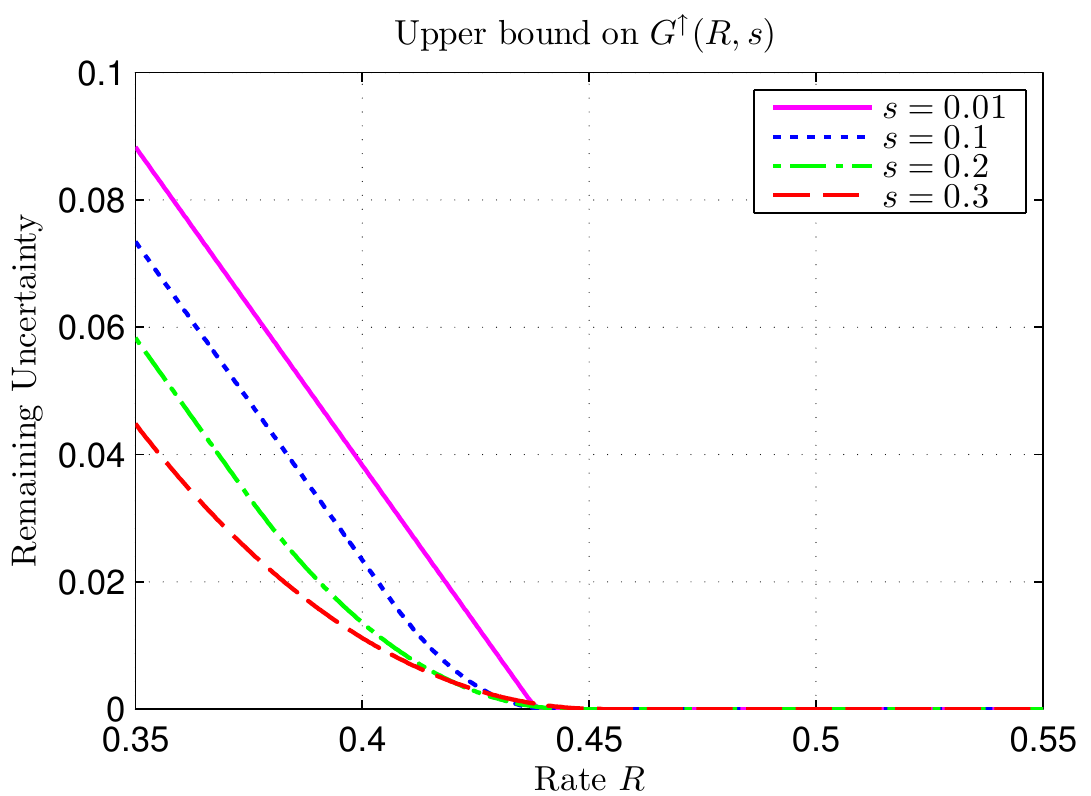} 
\end{tabular}
\caption{Illustration of the upper bounds on the remaining uncertainties $G (R,   s) $ and $G^\uparrow (R,   s) $ in \eqref{eqn:plus_max_res} and \eqref{eqn:plus_max_uni_str} respectively. All curves transition from a positive quantity to zero at the  conditional Shannon entropy $H(A|E|P_{AE})\approx 0.44$ nats.}
\label{fig:rm}
\end{figure}

In Figs.~\ref{fig:rm_simple} and~\ref{fig:rm}, we plot the upper bounds in  \eqref{eqn:min_uni_res}--\eqref{eqn:plus_max_uni_str}  for a correlated source $P_{AE} \in\calP( \{0,1\}^2)$ with $P_{AE}(0,0) = 0.7$ and $P_{AE}(0,1)=P_{AE}(1,0)= P_{AE}(1,1)=0.1$.   For the upper bounds in  \eqref{eqn:min_uni_res} and \eqref{eqn:min_gal_res}, we see from Fig.~\ref{fig:rm_simple} that the rates at which the curves  transition from a positive quantity to zero are clearly the  conditional R\'enyi entropies $H_{1-s}(A|E|P_{AE})$ and $H_{1-s}^\uparrow(A|E|P_{AE})$. In contrast, from  Fig.~\ref{fig:rm}, we observe that the rates at which the normalized remaining uncertainties  transition from positive quantities to zero are the same and are equal to the  conditional Shannon entropy $H(A|E|P_{AE})\approx 0.44$ nats.

\subsection{Optimal  Rates for Vanishing Remaining Uncertainties} \label{sec:opt_rates}
 The tightness of the bounds in Theorem \ref{thm:rem} is partially addressed in the following theorem where we are  concerned with the minimum compression rates $R$ such that the various normalized remaining uncertainties  tend to zero. 

\begin{figure}
\centering 
\begin{tabular}{ccc}
\includegraphics[width = .30\columnwidth]{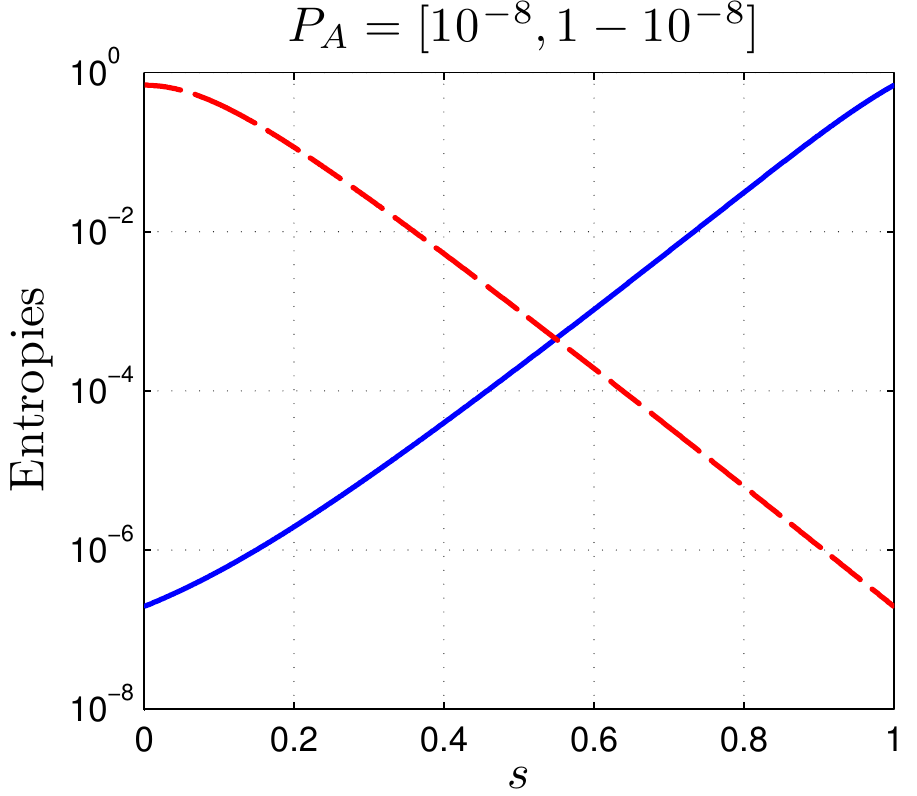}   & 
\includegraphics[width = .30\columnwidth]{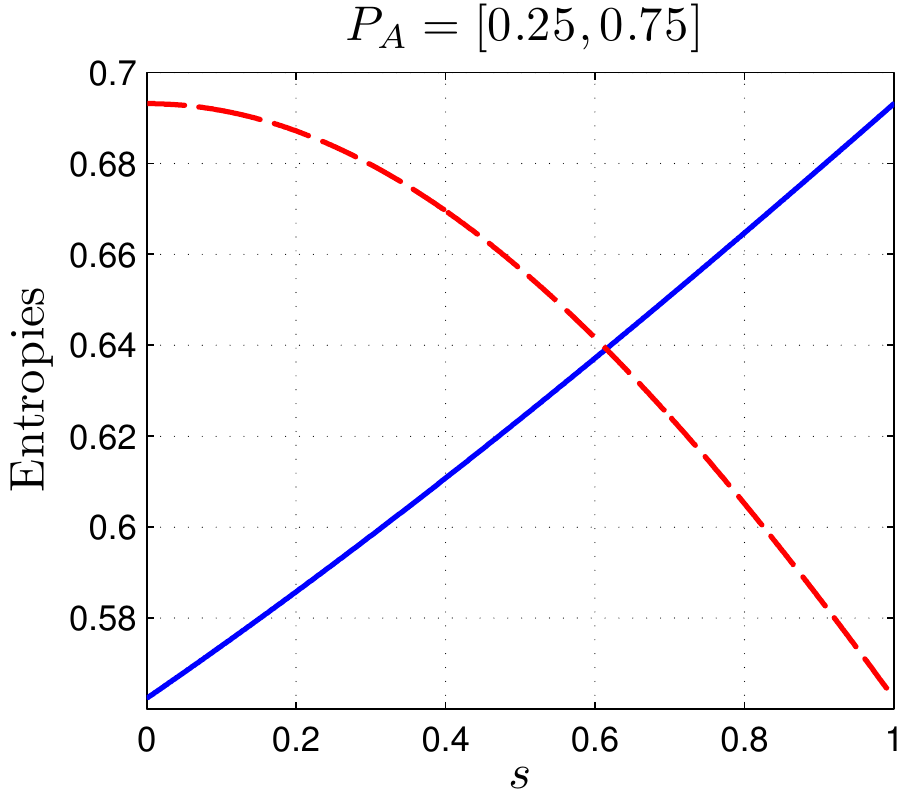}   & 
\includegraphics[width = .30\columnwidth]{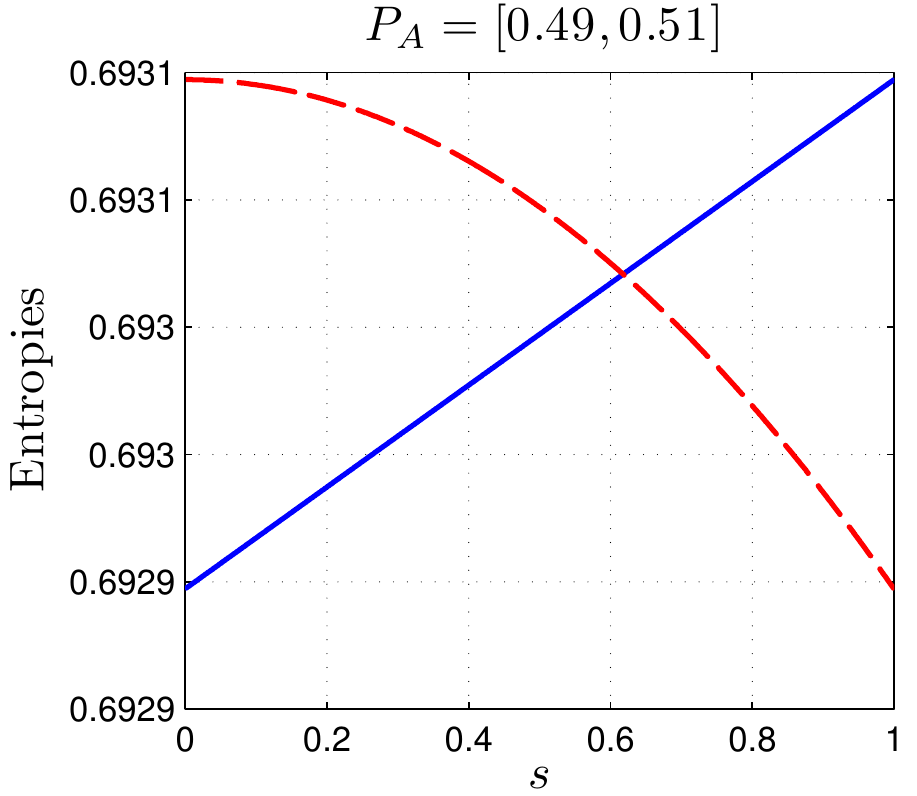}      
\end{tabular}
\caption{Plots of $H_{1-s}(A|P_{A})$ (solid) and  $H(A|P_A^{(s-1)})$ (dotted) for $P_A  \in \calP( \{ 0,1\})$ where $P_A(0) = 10^{-8}$ (left; almost deterministic), $P_A(0) = 0.25$ (middle) and $P_A(0)=0.49$ (right; almost uniform). For the three cases, $s_0(A|P_A) \approx  0.549$, $s_0(A|P_A)\approx 0.615$ and $s_0(A|P_A)\approx 0.618$ respectively.  }
\label{fig:ce}
\end{figure}
To state the next result succinctly, we require a few additional definitions. Let $P_A\in\calP(\calA)$ be a given distribution.  Let $\gamma (t) := t H_{1+t}(A|P_A) =-\log\sum_a P_A(a)^{1+t}$ and let $P_A^{(t)}(a):= P_A(a)^{1+t} \rme^{\gamma (t)}$ be a {\em tilted distribution}\footnote{$P_A^{(t)}(\fndot)$ is indeed a  valid distribution as $\sum_a P_A^{(t)}(a)=1$.} relative to $P_A$.   Define 
\begin{equation}
s_{0}(A|P_A):= \max\big\{ s \in [0,1] : H_{1-s}(A|P_{A}) \le  H(A|P_A^{(s-1)})  \big\} . \label{eqn:def_s0A}
\end{equation}
We claim that  $s_0(A|P_A)$ is always positive; this is because $s\mapsto  H_{1-s}(A|P_{A})$ and $s\mapsto H(A|P_A^{(s-1)}) $ are   continuous   and 
\begin{align}
H_{1-s}(A|P_A)  &= \left\{ \begin{array}{cc}
H(A|P_A) & s = 0\\
\log|\calA| & s =1 
\end{array}
 \right.  , \quad \mbox{and} \\
 H (A|P_A^{(s-1)})  &= \left\{ \begin{array}{cc}
\log|\calA|& s = 0\\
H(A|P_A)  & s =1 
\end{array}
 \right. .
\end{align}
If $A\sim P_A$ is not uniform on $\calA$, $s_0(A|P_A) \in (0,1)$. In fact  since $s\mapsto H_{1-s}(A|P_{A})$ and $s\mapsto H(A|P_A^{(s-1)})$ are monotonically increasing  and decreasing\footnote{Intuitively, $H(A|P_A^{(s-1)}) $ is monotonically decreasing because as $s$ increases, $P_A^{(s-1)}$ converges to a deterministic distribution, which has the lowest Shannon entropy $0$.}  respectively, $s_0(A|P_A) \in (0,1)$ can also be expressed as the {\em unique} solution to the equation $H_{1-s}(A|P_{A}) =  H(A|P_A^{(s-1)})$.  If 
 $A\sim P_A$ is uniform on $\calA$, $H_{1-s}(A|P_{A}) =  H(A|P_A^{(s-1)})$ for any $s\in [0,1]$ and as such, $s_0(A|P_A)=1$. 
 See Fig.~\ref{fig:ce} for   illustrations of these arguments.  Now, given $P_{AE} \in \calP(\calA\times\calE)$, define
\begin{equation}
s_0=s_{0}(A|E |P_{AE}) := \min\{ s_{0}(A|P_{A|E=e}) : e\in\calE \}. \label{eqn:def_s0}
\end{equation}
Clearly, by the preceding arguments and the fact that $\calE$ is a finite set, $s_0$ is  positive.

\begin{theorem}[Optimal   Rates for Vanishing Normalized Remaining Uncertainties]  \label{cor:threshold}
For each $n\in\bbN$, let the size   of the  range of $f_{X_n}$ be $M_n=\rme^{nR}$. 
Define the best-case    limiting  normalized remaining uncertainties  over all random hash functions  as
\begin{align}
\widetilde{G}(R,s) &:= \varliminf_{n\to\infty}\frac{1}{n}\inf_{f_{X_n} \in\calR} H_{1 + s} ( A^n | f_{X_n}(A^n) , E^n, X_n | P_{AE}^n\times P_{X_n} ),\quad\mbox{and} \label{eqn:Gdef_inf}\\
\widetilde{G}^\uparrow(R,s)&:= \varliminf_{n\to\infty}\frac{1}{n}\inf_{f_{X_n} \in\calR} H_{1 + s}^\uparrow ( A^n | f_{X_n}(A^n) , E^n, X_n | P_{AE}^n\times P_{X_n} ) \label{eqn:Gdef_up_inf} .
\end{align}
Also define the limiting  normalized remaining uncertainty  for strongly universal hash functions  $\barf_{X_n} : \calA^n\to\{1,\ldots, M_n\}$ as
\begin{equation}
\overline{G}(R,s)   := \varliminf_{n\to\infty}\frac{1}{n}H_{1 + s}( A^n | \barf_{X_n} (A^n), E^n, X_n | P_{AE}^n\times P_{X_n} )  .\label{eqn:Gdef_inf_str} 
\end{equation}
Now define the optimal compression rates 
\begin{align}
 T_s  &:= \inf\{ R \in\bbR: {G}(R,s) =0 \},\\
\widetilde{T}_s & := \inf\{ R\in\bbR: \widetilde{G}(R, s) =0 \},\\
 T_s^\uparrow  &:= \inf\{ R\in\bbR: {G}^\uparrow(R,s) =0 \},\\
\widetilde{T}^\uparrow_s & := \inf\{ R\in\bbR: \widetilde{G}^\uparrow(R, s) =0 \},\quad\mbox{and} \\
\overline{T}_s & := \inf\{ R\in\bbR: \overline{G}(R, s) =0 \}.
\end{align}
\begin{enumerate}
\item For $s\in [0,1]$, we have 
 \begin{align}
T_{-s} \le     H_{1-s}(A|E|P_{AE}) .  \label{eqn:optkey1}
  \end{align}
and for $s\in [ 0,s_0]$, we have 
\begin{equation}
\overline{T}_{-s}  \ge   H_{1-s}(A|E|P_{AE})    .\label{eqn:optkey1_s0}
\end{equation}
\item For $s\in ( 0,1)$, we have 
 \begin{align}
T_s =  \widetilde{T}_s  &=  H(A|E|P_{AE}). \label{eqn:optkey1a} 
  \end{align}
\item   For $s\in [0,1/2]$, we have 
  \begin{align}
 {T}^\uparrow_{-s}  = \widetilde{T}^\uparrow_{-s}   =  H_{1-s}^\uparrow(A|E|P_{AE}) ,  \label{eqn:optkey2} 
 \end{align}
and for $s\in [0,\infty)$, we have 
 \begin{align}
 {T}^\uparrow_{ s}  = \widetilde{T}^\uparrow_{ s}  =  H(A|E|P_{AE}). \label{eqn:optkey2a} 
 \end{align}
 \end{enumerate}
 \end{theorem}
The proof of this result is provided in Section \ref{sec:threshold_prf}. 

 For Part (1)  of the above result, unfortunately, we do not have a matching lower bound to $T_{-s} $. However, for $s\in [0,s_0]$, the bound in~\eqref{eqn:optkey1_s0} says that   restricted to the important class of strongly universal hash functions (e.g., the   ubiquitous random binning procedure~\cite{cover75}), the result in \eqref{eqn:optkey1} is tight as there is a matching lower bound. Hence, \eqref{eqn:optkey1_s0} serves as a ``partial converse''   to~\eqref{eqn:optkey1}. In other words,  \eqref{eqn:optkey1} is {\em tight with respect to the ensemble average}~\cite{Gallager73}  when the ensemble is chosen to be  a  strongly universal hash function. 
 
The equalities in \eqref{eqn:optkey1a}--\eqref{eqn:optkey2a} imply  that in the specified ranges of $s$, the optimal rates for the best-case remaining uncertainty over all hash functions and worst-case remaining uncertainty over all universal$_2$ hash functions are the same. It is interesting to 
observe that the   optimal rate  for the $-s$ case in~\eqref{eqn:optkey2} depends on $s \in [0,s_0]$ but the optimal rates for the $+s$ cases  in~\eqref{eqn:optkey1a} and~\eqref{eqn:optkey2a} do not. This is also clearly observed in Figs.~\ref{fig:rm_simple} and~\ref{fig:rm}.  

The proofs of the achievability parts (upper bounds) of these results follow directly from Theorem \ref{thm:rem}.  For the converse parts (lower bounds), we appeal to the method of types~\cite[Ch.~2]{Csi97}, the moments of type class enumerator method~\cite{merhav08, merhav14,  merhav_FnT, kaspi11,merhav13}, and the exponential strong converse for Slepian-Wolf coding~\cite[Theorem~2]{Oohama94}. We also exploit  a result by Fehr and Berens~\cite[Theorem~3]{fehr} concerning the monotonicity~\eqref{eqn:mono} and chain rule~\eqref{eqn:cr}   for Gallager form of the conditional R\'enyi entropy $H_{1-s}^\uparrow(A|E|P_{AE})$.

 \subsection{Exponential Rates of Decrease of Remaining Uncertainties}
Lastly, we consider the rate of exponential decrease of the various worst-case remaining uncertainties.


\begin{theorem}[Exponents of Remaining Uncertainties] \label{thm:exponents}
For each $n\in\bbN$, let the size of the  range of $f_{X_n}$ be $M_n=\rme^{nR}$. Fix a joint distribution $P_{AE}\in\calP(\calA\times\calE)$.  Define the exponents  of \eqref{eqn:Gdef} and \eqref{eqn:Gdef_up} as
\begin{align}
E(R,s) &:= \varliminf_{n\to\infty}-\frac{1}{n}\log \left[\sup_{f_{X_n} \in\calU_2} H_{1 + s} ( A^n | f_{X_n}(A^n) , E^n, X_n | P_{AE}^n\times P_{X_n} ) \right],\quad\mbox{and} \label{eqn:def_exp}\\
E^\uparrow(R,s)&:= \varliminf_{n\to\infty}-\frac{1}{n}\log\left[\sup_{f_{X_n} \in\calU_2} H_{1 + s}^\uparrow ( A^n | f_{X_n}(A^n) , E^n, X_n | P_{AE}^n\times P_{X_n} ) \right].\label{eqn:def_exp_gal}
\end{align}
These are the exponents  of the worst-case remaining uncertainties over all universal$_2$ hash functions. 
The following achievability statements hold:
\begin{enumerate}
\item For $s\in [0,1]$, we have  
\begin{equation}
E (R,  -s) \ge \left|\sup_{t \in (s,1) } t  \big(R - H_{1-t}(A|E|P_{AE}) \big)  \right|^+ , \label{eqn:exp1}
\end{equation}
 and for any $s \in [0,1/2]$, we have  
\begin{equation}
E^\uparrow (R,-s) \ge\left|\sup_{t \in (s,1/2) } \frac{t}{1-t} \left(R - H_{1-t}^\uparrow(A|E|P_{AE}) \right)  \right|^+ .   \label{eqn:exp2}
\end{equation}
\item For $s\in [0,\infty)$, we   have 
\begin{align}
E(R, s) \ge   \sup_{t \in (0,1 /2) } \frac{t}{1-t}\left(R - H_{1-t}^\uparrow (A|E|P_{AE}) \right)  , \label{eqn:exp3}
\end{align}
and 
\begin{align}
E^\uparrow(R, s)& \ge  \sup_{t \in (0,1/2) } \frac{t}{1-t}\left(R - H_{1-t}^\uparrow(A|E|P_{AE})\right) .  \label{eqn:exp4}
\end{align}
\end{enumerate}
\end{theorem}

Theorem \ref{thm:exponents} is proved in Section \ref{sec:prf_exp}. 

We observe from Proposition \ref{pr:err_exp} that the right-hand-sides  of  the bounds in Part (2), which can be shown to be non-negative for $R\ge H(A|E|P_{AE})$, are lower bounds on the optimal error exponent~\cite{Gal76,Kos77,Oohama94}   for the Slepian-Wolf~\cite{sw73} problem, denoted as $E^*_{\mathsf{SW}}(R)$. In fact, it can be inferred from Gallager's work~\cite{Gal76}  (or \cite[Problem~2.15(a)]{Csi97} for the $\calE = \emptyset$ case) that if we replace the domain of the optimization over $t$ from $(0,1/2)$  to $(0,1)$, the lower bounds in \eqref{eqn:exp3}--\eqref{eqn:exp4} are equal to $E^*_{\mathsf{SW}}(R)$ for a certain range of coding rates above $H(A|E|P_{AE})$. The reason why we obtain a potentially smaller exponent is because we consider the {\em worst-case} over all universal$_2$ hash functions $f_{X_n}\in\calU_2$ in the definitions of $E(R,s)$ and $E^\uparrow(R,s)$ in   \eqref{eqn:def_exp} and \eqref{eqn:def_exp_gal} respectively. For the Slepian-Wolf problem, we can choose the {\em best}  sequence of hash functions.


\begin{figure}
\centering
\begin{tabular}{cc}
\includegraphics[width = .475\columnwidth]{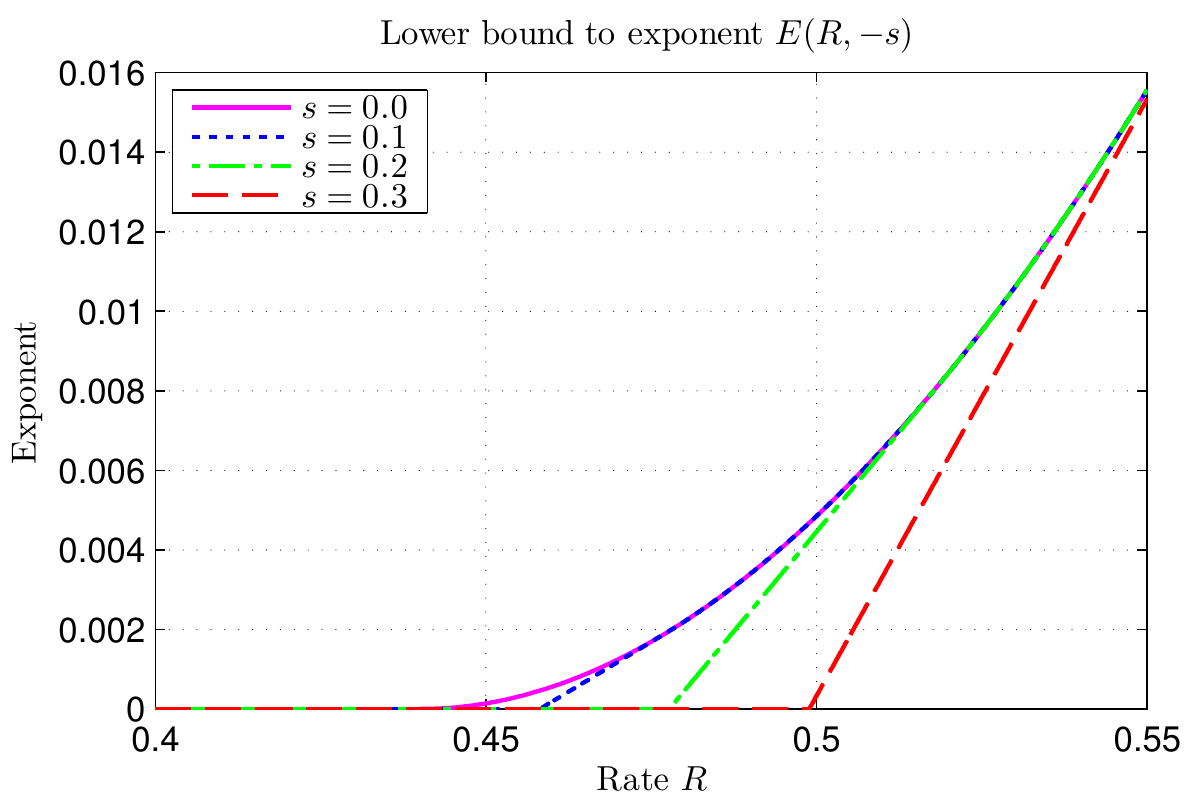} &
\includegraphics[width = .475\columnwidth]{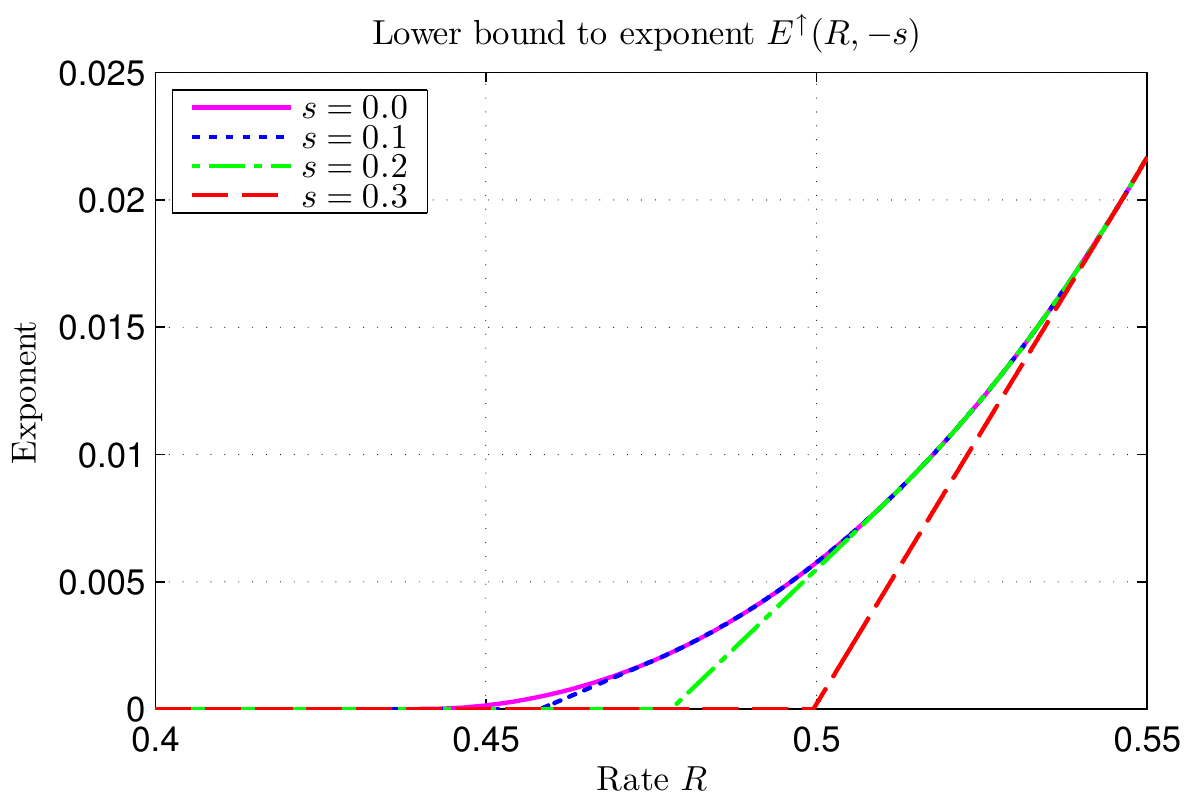} 
\end{tabular}
\caption{Illustration of the lower bounds on the exponents $E (R,  -s) $ and $E^\uparrow (R,  -s) $ in \eqref{eqn:exp1} and \eqref{eqn:exp2} respectively. The curves transition from $0$ to a positive quantity at $H_{1-s}(A|E|P_{AE})$ (left)  and $H_{1-s}^\uparrow(A|E|P_{AE})$ (right). }
\label{fig:exponents}
\end{figure}

In Fig.~\ref{fig:exponents}, we plot the lower bounds in  \eqref{eqn:exp1} and \eqref{eqn:exp2} for the same source $P_{AE}$  as in Figs.~\ref{fig:rm_simple} and \ref{fig:rm} in Section \ref{sec:asymp_ru}.  We note that the rates at which the lower bounds on the exponents transition from being zero to positive is given by $H_{1-s}(A|E|P_{AE})$ for \eqref{eqn:exp1} and $H_{1-s}^\uparrow(A|E|P_{AE})$ for \eqref{eqn:exp2}. The latter observation corroborates     \eqref{eqn:optkey2} of Theorem~\ref{cor:threshold}. Furthermore, as discussed in the previous paragraph, the $s=0$ case in the right plot of Fig.~\ref{fig:exponents} is a lower bound on $E^*_{\mathsf{SW}}(R)$. For this source, if we change the domain of optimization of $t$ from $(0,1/2)$  to $(0,1)$, the plot does not change (i.e., the optimal $t<1 /2$) so for rates in a small neighborhood above $H(A|E|P_{AE})\approx 0.44$ nats, the curve indeed traces out the optimal error exponent $E^*_{\mathsf{SW}}(R)$.

\section{Proof of Theorem~\ref{thm:rem}} \label{sec:prf_rem}
We   prove statements \eqref{eqn:min_uni_res}, \eqref{eqn:min_gal_res}, \eqref{eqn:plus_max_res}, and \eqref{eqn:plus_max_uni_str} in Subsections \ref{sec:dir_rem_1}, \ref{sec:dir_rem_2}, \ref{sec:dir_rem_3}, and \ref{sec:dir_rem_4} respectively.

\subsection{Proof of  \eqref{eqn:min_uni_res} in Theorem \ref{thm:rem}}\label{sec:dir_rem_1}
To prove the upper bound in \eqref{eqn:min_uni_res}, we use the one-shot bound in \eqref{10-15-30} in Lemma~\ref{lem:one-shot-direct} (Appendix \ref{app:one-shot}). We first assume that $ H_{1-s}(A|E|P_{AE}) -  R > 0$. In this case,  
\begin{equation}
M^{-s} \epsilon^s \rme^{s H_{1-s} (A^n|E^n|P_{AE}^n)}\ge 1, \label{eqn:suff_large0}
\end{equation}
   for $n$ sufficiently large.   Then the one-shot bound in  \eqref{10-15-30} implies that  for any $\epsilon$-almost universal$_2$ hash function $f_{X_n}$, 
\begin{align}
&H_{1-s}(A^n | f_{X_n}(A^n) , E^n,X_n |P_{AE}^n\times P_{X_n}) \nn\\*
&=\frac{1}{s}\log\bbE_{X_n} \left( \rme^{s H_{1-s} (A^n |f_{X_n} (A^n) ,E^n | P_{AE}^n ) } \right)\\
&\le \frac{1}{s}\log\left( 1+  \frac{\epsilon^s \rme^{s H_{1-s} (A^n|E^n|P_{AE}^n)}}{M^s}\right)\label{eqn:use_suff_large2} \\
&\le\frac{1}{s}\log \left( 2\cdot   \frac{\epsilon^s \rme^{s H_{1-s} (A^n|E^n|P_{AE}^n)}}{M^s}\right)\label{eqn:use_suff_large0} \\
&=\frac{1}{s}\log \left( 2\epsilon^s\right) +    n \big(  H_{1-s} (A|E|P_{AE}) - R  \big) ,  \label{eqn:use_suff_large1} 
\end{align}
where in \eqref{eqn:use_suff_large0} we used \eqref{eqn:suff_large0} and in \eqref{eqn:use_suff_large1}  we used the fact that the conditional R\'enyi entropy is additive for independent random variables, i.e.,  $ H_{1-s} (A^n|E^n|P_{AE}^n)=  nH_{1-s} (A |E |P_{AE} )$.   
Since this bound holds for all $\epsilon$-universal$_2$ hash functions  $f_{X_n} \in\calU_2$ (including $\epsilon=1$),    normalizing by $n$, taking the $\varliminf$, and appealing to the definition $G(R,-s)$ in \eqref{eqn:Gdef} establishes that $ G(R,-s)\le H_{1-s}(A|E|P_{AE})-R$ if $H_{1-s}(A|E|P_{AE})-R\ge 0$.

Now, when $H_{1-s}(A |E |P_{AE}) -  R \le  0$, we follow the steps leading to~\eqref{eqn:use_suff_large2} but use  $\log(1+t)\le t$  to establish that 
\begin{align}
H_{1-s}(A^n | f_{X_n}(A^n) , E^n,X_n |P_{AE}^n\times P_{X_n})   \le \frac{1}{s}\epsilon^s \cdot \rme^{ns (H_{1-s}(A|E|P_{AE})-R )} . \label{eqn:sec_case}
\end{align}
From \eqref{eqn:sec_case}, we conclude that if $H_{1-s}(A |E |P_{AE}) -  R  \le  0$, we have $G(R,-s) =   0$ (because $G(R,-s)$ cannot be negative). Since the two  bounds  in \eqref{eqn:use_suff_large1} and \eqref{eqn:sec_case} hold  for all sequences of universal$_2$ hash functions $f_{X_n} \in\calU_2$ (taking $\epsilon=1$ above), together they establish  \eqref{eqn:min_uni_res}.

\subsection{Proof of  \eqref{eqn:min_gal_res} in Theorem \ref{thm:rem}}\label{sec:dir_rem_2}
To prove the upper bound in \eqref{eqn:min_gal_res}, we use the one-shot bound in \eqref{10-15-31} in Lemma~\ref{lem:one-shot-direct} (Appendix \ref{app:one-shot}). Similarly, to the analysis in Section~\ref{sec:dir_rem_1}, we may consider two cases $H_{1-s}^\uparrow (A|E|P_{AE})-R >0$ or $H_{1-s}^\uparrow (A|E|P_{AE})-R \le 0$. We will only consider the former since the analysis of the latter parallels that in Section \ref{sec:dir_rem_1}. Under the former condition, we may assume that 
\begin{equation}
M^{- \frac{s}{1-s}} \epsilon^{  \frac{s}{1-s}}  \rme^{  \frac{s}{1-s} H_{1-s}^\uparrow(A^n|E^n|P_{AE}^n)} \ge 1 \label{eqn:suff_large}
\end{equation}
 for $n$ sufficiently large.  The one-shot bound in \eqref{10-15-31} implies that  for any $\epsilon$-almost universal$_2$ hash function $f_{X_n}$, 
\begin{align}
&H_{1-s}^\uparrow(A^n | f_{X_n}(A^n) , E^n,X_n |P_{AE}^n\times P_{X_n}) \nn\\* 
&=\frac{1-s}{s} \log\bbE_{X_n}\left(  \rme^{ \frac{s}{1-s}  H_{1-s}^\uparrow(A^n | f_{X_n}(A^n) , E^n  | P_{AE}^n ) }  \right)\\
&\le \frac{1-s}{s} \log \left( 1+  \frac{\epsilon^{  \frac{s}{1-s}}  \rme^{  \frac{s}{1-s} H_{1-s}^\uparrow(A^n|E^n|P_{AE}^n)}}{M^{\frac{s}{1-s}}} \right)\\
&\le \frac{1-s}{s} \log \left( 2\cdot \frac{\epsilon^{  \frac{s}{1-s}}  \rme^{  \frac{s}{1-s} H_{1-s}^\uparrow(A^n|E^n|P_{AE}^n)}}{M^{\frac{s}{1-s}}}\right) \label{eqn:use_suff_large}\\
&=\frac{1-s}{s}\log \left(2 \epsilon^{  \frac{s}{1-s}} \right) + n  \left(H_{1-s}^\uparrow(A|E|P_{AE}) - R \right),
\end{align}
where in \eqref{eqn:use_suff_large} we used \eqref{eqn:suff_large}.  
Since this bound holds for all sequences of universal$_2$ hash functions $f_{X_n} \in\calU_2$ (taking $\epsilon=1$ above),  normalizing by $n$ and appealing to the definition $G^\uparrow(R,-s)$ in \eqref{eqn:Gdef_up}, we establish  the upper bound in \eqref{eqn:min_gal_res}.
 
\subsection{Proof of \eqref{eqn:plus_max_res} in Theorem \ref{thm:rem}}\label{sec:dir_rem_3}
To prove the upper bound in  \eqref{eqn:plus_max_res}, we will resort to the one-shot bound in \eqref{10-16-4} in Lemma~\ref{lem:one-shot-direct2} (Appendix \ref{app:one-shot}). We first observe by Cram\'er's theorem~\cite[Section 2.2]{Dembo} that 
\begin{equation}
\lim_{n\to\infty} -\frac{1}{n}\log P_{AE}^n \left\{ (\ba,\be) : P_{A|E}^n(\ba|\be)\ge\epsilon \rme^{-nR}\right\}=\sup_{t\ge 0}  \big\{t( H_{1+t}(A|E|P_{AE})-R) \big\}.  \label{eqn:cram_bd1}
\end{equation}
This is because the cumulant generating function of the random variable $\log P_{A|E}(A|E)$ where $(A,E)$ is distributed as $P_{AE}$ is 
\begin{equation}
\log \bbE \left[\rme^{t  \log P_{A|E}(A|E) } \right] = \log\sum_{e} P_E(e) \sum_a P_{A|E}(a|e)^{1+t} = -tH_{1+t}(A|E|P_{AE}).
\end{equation}
Next, we apply  a generalization  of   Cram\'er's theorem concerning arbitrary finite non-negative measures\footnote{The standard Cram\'er's theorem~\cite[Section 2.2]{Dembo} (or Sanov's theorem~\cite[Section 2.1]{Dembo}) is a  large-deviations result concerning the exponent  of $P^n(\calB)$ where $P$ is a {\em probability} measure and $\calB$ is an event in the sample space $\Omega$.  If $P$ is not necessarily a probability measure but a finite non-negative measure (as it is in our applications), say $\mu$,    Cram\'er's theorem  clearly also applies by defining the  new {\em probability} measure $\calB\mapsto \widetilde{P}(\calB) := \mu(\calB)/\mu(\Omega)$.  \label{fn:cramer} } (not necessarily probability measures) 
to the sequence of  random variables $-\log P_{A|E}^n(A^n|E^n) =\sum_{i=1}^n -\log P_{A|E}(A_i | E_i)$ under  the sequence of non-negative finite  joint measures $\calB\mapsto  \sum_{ (\ba,\be) \in \calB} P_{AE}^n(\ba,\be) P_{A|E}^n(\ba|\be)^s $ to establish  that 
\begin{align}
&\lim_{n\to\infty} -\frac{1}{n}\log \sum_{(\ba,\be) :   P_{A|E}^n(\ba|\be) <  \epsilon \rme^{-nR}} P_{AE}^n(\ba,\be) P_{A|E}^n(\ba|\be)^s \rme^{snR}\nn\\*
&\qquad=\left\{ \begin{array}{cc}
s(H_{1+s}(A|E|P_{AE}) - R) & R \le \hatR_s \\
\max_{t\in [0,s]} t(H_{1+t}(A|E|P_{AE}) - R) & R \ge\hatR_s    
\end{array} \right.\label{eqn:cram_bd2}  .
\end{align}
The  statement in \eqref{eqn:cram_bd2}   holds  because the relevant cumulant generating function is 
\begin{align}
\tau_s(t)& = \log\sum_{a,e} P_{AE}(a,e) P_{A|E}(a|e)^s \rme^{ - t\log P_{A|E}(a|e) }  \\
&= - (s-t) H_{1+(s-t)} (A|E|P_{AE}), \label{eqn:deftau}
\end{align}
from the definition of the conditional R\'enyi entropy in \eqref{eqn:renyi_ent_Q}. 
Hence,
\begin{align}
\Gamma_s&:=\lim_{n\to\infty} -\frac{1}{n}\log \sum_{(\ba,\be) :   P_{A|E}^n(\ba|\be) <  \epsilon \rme^{-nR}} P_{AE}^n(\ba,\be) P_{A|E}^n(\ba|\be)^s \\
&= \sup_{t\ge 0} \left\{ tR-\tau_s(t) \right\}\\
&= \sup_{t\ge 0} \left\{  tR+ (s-t)H_{1 +(s-t)}(A|E|P_{AE}) \right\}. \label{eqn:cramer_res}
\end{align}
By differentiating the objective function in \eqref{eqn:cramer_res}, we see that if $R\le\hatR_s =\frac{\rmd}{\rmd t}t H_{1+t}\big|_{t=s}$ (cf.\ the definition of the critical rate in \eqref{eqn:crit_rate1}), the optimal solution is attained at  $t^*=0$  (recall that $t\mapsto tH_{1+t}$ is concave so $s\mapsto\hatR_s$ is decreasing) and so $\Gamma_s= sH_{1 +s}(A|E|P_{AE})$,  leading to the first clause in \eqref{eqn:cram_bd2}. Conversely, when $R > \hatR_s$, the optimal solution is attained at  $t^* > 0$. This leads to the second clause   on the right-hand-side of  \eqref{eqn:cram_bd2} because the left-hand-side of  \eqref{eqn:cram_bd2} is now
\begin{align}
\Gamma_s - sR  &= \sup_{t \ge 0 }  \left\{  (t-s)R   + (s-t)H_{1 +(s-t)}(A|E|P_{AE}) \right\}  \\
& =\max_{t  \in [0,s]}  \left\{ t H_{1 +t }(A|E|P_{AE}) -t  R\right\}.
\end{align}
Since \eqref{eqn:cram_bd1} is not smaller than \eqref{eqn:cram_bd2}, the latter dominates. Now using the one-shot bound in \eqref{10-16-4} in Lemma~\ref{lem:one-shot-direct2} we see that for any sequence of $\epsilon$-almost universal$_2$ hash functions $f_{X_n}$,
\begin{align}
&\varlimsup_{n\to\infty}\frac{1}{n}H_{1+s} (A^n| f_{X_n}(A^n) , E^n, X_n | P_{AE}^n\times P_{X_n} )  \nn\\*
&=\varlimsup_{n\to\infty} -\frac{1}{ns}\log \bbE_{X_n}  \left[ \rme^{-sH_{1+s} ( A^n| f_{X_n}(A^n) , E^n| P_{AE}^n )} \right]\\
&\le \varlimsup_{n\to\infty} -\frac{1}{ns}\log\bigg[ 2^{-s} \sum_{ (\ba,\be) : P_{A|E}^n(\ba|\be) \ge  \epsilon \rme^{-nR}} P_{AE}^n(\ba,\be)  \nn\\*
&\qquad\qquad\qquad+ 2^{-s}  \sum_{ (\ba,\be) : P_{A|E}^n(\ba|\be)  <   \epsilon \rme^{-nR}} P_{AE}^n(\ba,\be) P_{A|E}^n(\ba|\be)^s \rme^{snR}\bigg]\\
&= \left\{ \begin{array}{cc}
  H_{1+s}(A|E|P_{AE}) - R  & R \le\hatR_s \\
\max_{t\in [0,s]} \frac{t}{s}(H_{1+t}(A|E|P_{AE}) - R) & R \ge\hatR_s
\end{array} \right.  .
\end{align}
Since this bound holds for all sequences of universal$_2$ hash functions $f_{X_n} \in\calU_2$ (taking $\epsilon=1$ above), we have  established the upper bound in~ \eqref{eqn:plus_max_res}.

\subsection{Proof of \eqref{eqn:plus_max_uni_str} in Theorem \ref{thm:rem}} \label{sec:dir_rem_4}
We now prove the upper bound in \eqref{eqn:plus_max_uni_str}. For this purpose, we use the one-shot bound \eqref{10-16-4b} in Lemma~\ref{lem:one-shot3} (Appendix \ref{app:one-shot}). We employ Cram\'er's theorem \cite[Section 2.2]{Dembo}  with the sequence of random variables $\log P_{A|E}^n (\ba|\be)( \sum_{\tba} P_{A|E}^n(\tba|\be)^{1+s})^{-\frac{1}{1+s}}$ (where $\tba=(\tila_1,\tila_2,\ldots, \tila_n) \in\calA^n$) under the sequence of joint distributions $P_{AE}^n (\ba,\be)$. We claim that
\begin{align}
&\lim_{n\to\infty}-\frac{1}{n}\log\sum_{\be}P_{E}^n(\be)  \sum_{\ba : P_{A|E}^n(\ba|\be)^{1+s} \ge\epsilon \rme^{-nR} \sum_{\tba} P_{A|E}^n(\tba|\be)^{1+s}} P_{A|E}^n(\ba|\be) \nn\\* 
& =\lim_{n\to\infty}-\frac{1}{n}\log\sum_{ (\ba,\be) :    P_{A|E}^n(\ba|\be)  ( \sum_{\tba} P_{A|E}^n(\tba|\be)^{1+s})^{-\frac{1}{1+s}} \ge\epsilon \rme^{-n \frac{R}{1+s}}   }P_{AE}^n(\ba,\be) \\
&=\max_{t\ge 0 }\frac{t}{1+s} (H_{1  + t|1+s}(A|E|P_{AE}) - R). \label{eqn:ge1}
\end{align}
Let us justify the claim in \eqref{eqn:ge1} carefully. The derivation here is similar to that in~\eqref{eqn:deftau}--\eqref{eqn:cramer_res} and involves calculating the relevant cumulant generating function
\begin{align}
\tau_s(t) & :=\log\sum_{a,e} P_{AE}(a,e) \exp\left(  t \log  \left[ P_{A|E}(a|e) \bigg( \sum_{\tila} P_{A|E}(\tila|e)^{1+s} \bigg)^{-\frac{1}{1+s}}  \right] \right)\\
&=\log\sum_e P_E(e) \left( \sum_a P_{A|E}(a|e)^{1 + t} \right) \left( \sum_{\tila} P_{A|E}(\tila|e)^{1+s} \right)^{- \frac{t}{1+s}}\\
&= -\frac{t}{1+s}H_{1 + t|1+s} (A|E|P_{AE}),  
\end{align}
where the last step results from the definition of the two-parameter conditional R\'enyi entropy in  \eqref{eqn:two_param}. By an application of Cram\'er's theorem, the corresponding exponent is~\eqref{eqn:ge1}.

In addition, we  apply the generalized version of  Cram\'er's theorem (see footnote \ref{fn:cramer}) 
 to compute another large deviations quantity. Consider the sequence of  random variables  $-\log P_{A|E}^n(\ba|\be)( \sum_{\tba} P_{A|E}^n(\tba | \be)^{1+s})^{ -\frac{1}{1+s} }$ distributed according to the sequence of non-negative  finite joint measures $\calB\mapsto \sum_{(\ba,\be)\in \calB}P_{E}^n(\be)  (P_{A|E}^n (\ba | \be))^{1+s} ( \sum_{\tba} P_{A|E}^n (\tba|\be)^{1+s})^{ - \frac{s}{1+s}}$. We claim that  the exponent can be calculated to be
\begin{align}
&\lim_{n\to\infty}-\frac{1}{n}\log \Bigg[ \rme^{ n\frac{s}{1+s}R } \sum_{\be} P_E^n(\be) \nn\\*
&\qquad\qquad\times \sum_{ \ba : P_{A|E}^n(\ba|\be)^{1+s} < \epsilon \rme^{-nR} \sum_{\tba} P_{A|E}^n(\tba|\be)^{1+s}}  P_{A|E}^n(\ba|\be)^{1+s}\left(  \sum_{\tba} P_{A|E}^n(\tba|\be)^{1+s}\right)^{-\frac{s}{1+s}}\Bigg]\label{eqn:ge20} \\
&=\left\{ \begin{array}{cc}
\frac{s}{1+s} (H_{1+s}^\uparrow(A|E|P_{AE}) - R) & R \le \hatR_s^\uparrow  \\
\max_{t\in [0,s]}\frac{t}{1+s} (H_{1+t|1+s} (A|E|P_{AE} )- R) & R > \hatR_s^\uparrow
\end{array} \right. \label{eqn:ge2} .
\end{align}
Let us justify the claim in \eqref{eqn:ge2} carefully.  This    step follows because the   relevant cumulant generating function is 
\begin{align}
\tau_s(t)  &:= \log \sum_e P_E(e) \sum_a P_{A|E}(a|e)^{1+s} \left(\sum_{\tila} P_{A|E}(\tila|e)^{1+s} \right)^{  -\frac{s}{1+s} } \nn\\
&\qquad \times\exp\left( -t \log \bigg[ P_{A|E}(a|e) \Big(  \sum_{\tila} P_{A|E}(\tila|s)^{1+s}   \Big)^{ - \frac{1}{1+s}}\bigg] \right)\\
&= \log\sum_e P_E(e) \sum_a P_{A|E}(a|e)^{1+s-t} \left(\sum_{\tila} P_{A|E}(\tila|e)^{1+s} \right)^{   \frac{t-s}{1+s} }.
\end{align}
Thus~\eqref{eqn:ge20} reduces to
\begin{align}
&\max_{t \ge 0 }  \left\{-  \tau_s(t)-\frac{s-t}{1+s}R \right\}\nn\\*
&=\max_{t \le s }   \left\{-\log\sum_e P_E(e) \sum_a P_{A|E}(a|e)^{1+t} \left(\sum_{\tila} P_{A|E}(\tila|e)^{1+s} \right)^{   -\frac{t}{1+s} }- \frac{t}{1+s}R \right\}\\
&=\max_{t \le s} \left\{\frac{t}{1+s}H_{1+t|1+s}(A|E|P_{AE}) - \frac{t}{1+s}R\right\},  \label{eqn:t_opt}
\end{align}
where the last step follows from the definition of the two-parameter conditional R\'enyi entropy in  \eqref{eqn:two_param}. Now from the definition of the critical rate $\hatR_s^\uparrow$ in \eqref{eqn:crit_rate2} and the fact that $H_{1+s | 1+s} = H_{1+s}^\uparrow$, we know that if $R\le \hatR_{s}^\uparrow$, the maximization in  \eqref{eqn:t_opt} is attained $t=s$, resulting in the first case in \eqref{eqn:ge2}. Conversely, the second case results from   $R > \hatR_{s}^\uparrow$ where the domain of $t$ is $[0,s]$ since the eventual exponent cannot be negative. This proves~\eqref{eqn:ge2}.

Since \eqref{eqn:ge2} is not greater than \eqref{eqn:ge1}, the former dominates the exponential behavior of $G^\uparrow(R,s)$, and so plugging these evaluations into  the one-shot bound in~\eqref{10-16-4b} which holds     for any   $\epsilon$-almost universal$_2$ hash function $f_{X_n}$,
\begin{align}
&\varlimsup_{n\to\infty}\frac{1}{n}H_{1+s}^\uparrow (A^n| f_{X_n}(A^n) , E^n, X_n | P_{AE}^n\times P_{X_n} )  \nn\\*
&= \varlimsup_{n\to\infty}- \frac{1+s}{ns}\log \bbE_{X_n}\left[ \rme^{ -\frac{s}{1+s}H_{1+s}^\uparrow(A^n | f_{X_n}(A^n) , E^n | P_{AE}^n )} \right] \\
&\le \varlimsup_{n\to\infty}- \frac{1+s}{ns}\log \Bigg[   2^{-\frac{s}{1+s} }\sum_{\be}P_{E}^n(\be)  \sum_{\ba : P_{A|E}^n(\ba|\be)^{1+s} \ge\epsilon \rme^{-nR} \sum_{\tba} P_{A|E}^n(\tba|\be)^{1+s}} P_{A|E}^n(\ba|\be)  \nn\\*
&  \qquad +  2^{-\frac{s}{1+s} } \epsilon^{-\frac{s}{1+s}} \rme^{ n\frac{s}{1+s}R }\sum_{\be} P_E^n(\be)  \nn\\*
&\qquad\qquad \times \sum_{ \ba : P_{A|E}^n(\ba|\be)^{1+s} < \epsilon \rme^{-nR} \sum_{\tba} P_{A|E}^n(\tba|\be)^{1+s}}   P_{A|E}^n(\ba|\be)^{1+s}\left(  \sum_{\tba} P_{A|E}^n(\tba|\be)^{1+s}\right)^{-\frac{s}{1+s}}  \Bigg]\\
&=\left\{ \begin{array}{cc}
   H_{1+s}^\uparrow(A|E|P_{AE} )- R   & R \le \hatR_s^\uparrow \\
\max_{t\in [0,s]}\frac{t}{ s} (H_{1+t|1+s  } (A|E|P_{AE}) - R) & R >\hatR_s^\uparrow 
\end{array} \right. .
\end{align}
Since this bound holds for all sequences of universal$_2$ hash functions $f_{X_n} \in\calU_2$ (taking $\epsilon=1$ above), we have  established the upper bound in~\eqref{eqn:plus_max_uni_str}.
\section{Proof of Theorem \ref{cor:threshold}} \label{sec:threshold_prf}
The bounds on the optimal compression rates corresponding to the conditional R\'enyi entropy and Gallager form of the conditional R\'enyi entropy  are proved in Subsections \ref{sec:prf_optkey1} and  \ref{sec:prf_optkey2} respectively. 
\subsection{Proofs of  \eqref{eqn:optkey1}, \eqref{eqn:optkey1_s0}, and \eqref{eqn:optkey1a}} \label{sec:prf_optkey1} 
\begin{proof}
Recall the definitions of  the optimal rates $T_s$,   $\widetilde{T}_s$, and $\overline{T}_s$ in Theorem \ref{cor:threshold}.   
Since $\widetilde{G}(R,s)\le{G}(R, s)$,  and both functions are monotonically non-increasing in $R$, it holds that $\widetilde{T}_s\le T_s$ for all $s$. 
 
 First, we prove the upper bounds to  $T_{-s}$ and $T_{s}$ in Section \ref{sec:ub}; next we prove the lower bound  to $\overline{T}_{-s}$ in Section \ref{sec:lb1}; and finally we prove  the lower bound to $\widetilde{T}_s$ in Section \ref{sec:lb2}. 
 For the $+s$ case, the upper and lower bounds match for all $s\in (0,1)$ and so this proves \eqref{eqn:optkey1a}. 

\subsubsection{Upper Bounds} \label{sec:ub}
We refer to the statement in \eqref{eqn:min_uni_res}. We observe that if $R\ge H_{1-s}(A|E|P_{AE})$, $G(R,-s)=0$ since $G(R,-s)$ is upper bounded by $|H_{1-s} (A|E |P_{AE}) - R|^+$ and $G(R,-s)$ is non-negative.  Hence, $T_{-s}\le H_{1-s}(A|E|P_{AE})$ for all $s\in [0,1]$.   This proves \eqref{eqn:optkey1}.

Next we refer to the statement in  \eqref{eqn:plus_max_res}. If $R\ge H(A|E|P_{AE})$, we know from the monotonically decreasing nature of $s\mapsto H_{1+s}(A|E|P_{AE})$ that $H_{1+t} (A|E|P_{AE})-R$ is non-positive for  $t\in [0,s]$. Thus, the optimal $t$   in the optimization in $\max_{t\in [0,s]}\frac{t}{s}( H_{1+t} (A|E|P_{AE})-R)$ is $t^*=0$  and consequently, the optimal objective value is also $0$. On the other hand, for $R \in [\hatR_s,H(A|E|P_{AE}))$, the optimal $t^* \in (0,s]$ and so the the optimal objective value is positive. We conclude for $s\in [0,\infty)$ that  the optimal key generation rate is upper bounded by the  conditional Shannon entropy $H (A|E|P_{AE})$. In summary, we conclude that 
$
T_{s}  \le H (A|E|P_{AE})$ for all $s \in [0,\infty)$, proving the upper bound for \eqref{eqn:optkey1a}.

\subsubsection{Lower Bound to $\overline{T}_{-s}$ }\label{sec:lb1}
We  now  consider {\em strongly universal   hash functions}~\cite{Wegman81} (cf.\ Definition~\ref{def:strong_has}) and  $s \in [0,s_0]$, where $s_0=s_0(A|E|P_{AE})$ is defined in \eqref{eqn:def_s0}. More precisely, we shall show that  for the sequence of   strongly universal   hash functions $\{\barf_{X_n} : \calA^n\to \{1,\ldots, \rme^{nR}\}\}_{n\in\bbN}$ and any $s\in [0,s_0]$, we have 
\begin{equation}
\overline{G}(R, -s)  =\varliminf_{n\to\infty}\frac{1}{n} H_{1-s}( A^n | \barf_{X_n} (A^n), E^n, X_n | P_{AE}^n\times P_{X_n} ) \ge | H_{1-s}(A|E|P_{AE}) - R|^+ ,\label{eqn:strong_lb}
\end{equation}
immediately implying that $\overline{T}_{-s}\ge H_{1-s}(A|E|P_{AE})$. 
In fact, for this range of $s$, not only is it true that the minimum value of $R$ such that $\overline{G}(R,-s)=0$ coincides with that in \eqref{eqn:optkey1}, the bound in \eqref{eqn:strong_lb} serves as a {\em tight} lower bound to the achievability (upper) bound for $G(R,-s)$ in \eqref{eqn:min_uni_res} (at least for strongly universal   hash functions). In the following, we make use heavy use  of the method of types; relevant notation is summarized in Section~\ref{sec:types}.

For ease of exposition, we first consider the case in which $|\calE|=1$ or equivalently, $\calE=\emptyset$. Subsequently, we generalize our result to the general case in which $|\calE|>1$. Starting with the one-shot bound in  
 \eqref{eqn:derivation_H_1minuss} (in Lemma \ref{lem:one-shot-direct}), we have 
\begin{align}
&\bbE_{X_n} \left[ \rme^{s H_{1-s} (A^n | \barf_{X_n}(A^n), X_n | P_{A}^n  ) } \right]\nn\\*
& = \bbE_{X_n}  \left[ \sum_{\ba} P_A^n(\ba)^{1-s} \left( \sum_{ \tba \in \barf_{X_n}^{-1}  (\barf_{X_n} (\ba))} P_A^n (\tba)\right)^s \right]\\
 & \ge 2^{s-1}  \bbE_{X_n } \left[  \sum_{\ba}  P_A^n(\ba)^{1-s} \left(  P_A^n(\ba)^s +\bigg[ \sum_{ \tba \in \barf_{X_n}^{-1}  (\barf_{X_n} (\ba)) \setminus\{\ba\}}  P_A^n (\tba)\bigg]^s \right)  \right] \label{eqn:apply_jens} \\
  &\doteq \bbE_{X_n } \left[  \sum_{\ba}  P_A^n(\ba)^{1-s} \left(  P_A^n(\ba)^s +\bigg[ \sum_{Q \in\calP_n(\calA)} \sum_{ \tba \in\calT_Q\setminus\{\ba\}  : \barf_{X_n} (\ba)=\barf_{X_n} (\tba)  }  P_A^n (\tba)\bigg]^s \right)  \right] \label{eqn:split_typ} \\ 
 &\doteq  \bbE_{X_n } \left[  \sum_{\ba}  P_A^n(\ba)^{1-s} \left(  P_A^n(\ba)^s +\bigg[ \max_{Q \in\calP_n(\calA)} \sum_{ \tba \in\calT_Q\setminus\{\ba\}  : \barf_{X_n} (\ba)=\barf_{X_n} (\tba)  }  P_A^n (\tba)\bigg]^s \right)  \right] \label{eqn:split_typ0} \\ 
  &=   \sum_{\ba}  P_A^n(\ba)^{1-s} \left(  P_A^n(\ba)^s + \bbE_{X_n}  \left\{\max_{Q \in\calP_n(\calA)}  \bigg[  \sum_{ \tba \in\calT_Q\setminus\{\ba\}  : \barf_{X_n} (\ba)=\barf_{X_n} (\tba)  }  P_A^n (\tba)\bigg]^s \right\} \right)    \label{eqn:split_typ1} \\
        &\doteq   \sum_{\ba}  P_A^n(\ba)^{1-s} \left(  P_A^n(\ba)^s + \sum_{Q \in\calP_n(\calA)} \bbE_{X_n} \left\{  \bigg[  \sum_{ \tba \in\calT_Q\setminus\{\ba\}  : \barf_{X_n} (\ba)=\barf_{X_n} (\tba)  }  P_A^n (\tba)\bigg]^s  \right\}\right)    \label{eqn:split_typ3}  ,   
\end{align}  
where in \eqref{eqn:apply_jens}  we used the bound $(b+c)^s \ge 2^{s-1}(b^s + c^s)$ for $b,c\ge 0$ and $s\in [0,1]$ (a consequence of Jensen's inequality applied to the concave function $t\mapsto t^{s}$ for $s\in [0,1]$), in \eqref{eqn:split_typ}, we split the inner sum into $n$-types on $\calA$, and in~\eqref{eqn:split_typ0}  and~\eqref{eqn:split_typ3} we used the fact that there are  polynomially many types so we can interchange sums over types with maximums over types and vice versa. This derivation is similar to \cite[Eqn.~(20)]{merhav08}.  

Now, we we assume that $R<H_{1-s}(A|P_A)$ and also that $s\le s_0(A|P_A)$. The latter assumption means that $H_{1-s}(A|P_A)\le H(A|P_A^{(s-1)})$ (see Section \ref{sec:opt_rates}).  In this case, we may use  the bound in~\eqref{eqn:case1_dom1}   in Lemma~\ref{lem:exp} (in Appendix~\ref{app:exp})  to lower bound the (inner) sum over   expectations in~\eqref{eqn:split_typ3}. We have  
\begin{align}
&\sum_{Q \in\calP_n(\calA)} \bbE_{X_n } \left\{ \bigg[\sum_{ \tba \in\calT_Q\setminus\{\ba\}  : \barf_{X_n} (\ba)=\barf_{X_n} (\tba)  }  P_A^n (\tba)\bigg]^s \right\} \nn\\*
&\dotgeq\sum_{Q \in\calP_n(\calA)} \left\{\bbE_{X_n }\bigg[\sum_{ \tba \in\calT_Q\setminus\{\ba\}  : \barf_{X_n} (\ba)=\barf_{X_n} (\tba)  } P_A^n(\tba) \bigg] \right\}^{s}  \label{eqn:use_lower_bds}\\
&\ge   \left\{\sum_{Q \in\calP_n(\calA)}\bbE_{X_n }\bigg[  \sum_{ \tba \in\calT_Q\setminus\{\ba\}  : \barf_{X_n} (\ba)=\barf_{X_n} (\tba)  }  P_A^n (\tba)\bigg]\right\}^s\label{eqn:gal_bds} \\
& =  \left\{ \bbE_{X_n }\bigg[  \sum_{ \tba\ne \ba   : \barf_{X_n} (\ba)=\barf_{X_n} (\tba)  }  P_A^n (\tba)\bigg]\right\}^s\label{eqn:gal_bds1} \\
&= \left\{  \sum_{\tba\ne \ba} P_{A}^n(\tba)  \Pr \left\{ \barf_{X_n} (\ba)=\barf_{X_n} (\tba) \right\}  \right\}^s \\
&= \left\{ \rme^{-nR} \sum_{\tba\ne \ba} P_A^n(\tba)\right\}^s \label{eqn:hash_prop}  ,
\end{align}
where   \eqref{eqn:use_lower_bds} uses \eqref{eqn:case1_dom1}  in Lemma \ref{lem:exp}, and \eqref{eqn:gal_bds} uses the inequality $\sum_i b_i^s \ge (\sum_ib_i)^s$ which holds for non-negative $b_i$ and $s \in [0,1]$ \cite[Problem~4.15(f)]{gallagerIT} and \eqref{eqn:hash_prop}  uses the fact that $\barf_{X_n}$ is a  strongly universal   hash function (cf.\ Definition~\ref{def:strong_has}).  Substituting~\eqref{eqn:hash_prop} into~\eqref{eqn:split_typ3}, we obtain
\begin{align}
&\bbE_{X_n} \left[ \rme^{s H_{1-s} (A^n | \barf_{X_n}(A^n), X_n | P_{A}^n  ) } \right] \nn\\*
& \dotgeq  \sum_{\ba}  P_A^n(\ba)^{1-s} \left(  P_A^n(\ba)^s  + \left\{ \rme^{-nR} \sum_{\tba\ne \ba} P_A^n(\tba)\right\}^s \right) \\
& = 1 + \rme^{-nsR} \sum_{\ba}  P_A^n(\ba)^{1-s} \big\{ 1- P_{A}^n(\ba)\big\}^s\\
&\doteq 1+ \exp \big(  ns  [ H_{1-s}(A|P_{A}) -  R] \big)  \label{eqn:lower_bd_H2}  \\*
&\doteq \exp\big(  ns  [ H_{1-s}(A|P_{A}) -  R] \big) \label{eqn:lower_bd_H} ,
\end{align}
where  in \eqref{eqn:lower_bd_H2} we used the fact that $\max_{\ba} P_A^n(\ba) \le \frac{1}{2}$ (say)    for $n$ large enough so $1\ge\{ 1- P_{A}^n(\ba)\}^s\ge 1-P_A^n(\ba)^s \ge 1-(\frac{1}{2})^s>0$, and  
in \eqref{eqn:lower_bd_H}, we used the condition $H_{1-s}(A|P_A)>R$ so the expression in \eqref{eqn:lower_bd_H2}  is exponentially large. 

In the other case, when $R\ge H_{1-s}(A|P_A)$, we simply lower bound  the sum over types term in \eqref{eqn:split_typ3} by $0$ and hence, the entire expression in \eqref{eqn:split_typ3} can be lower bounded  by $1$. Thus, we conclude that 
\begin{equation}
 \bbE_{X_n} \left[ \rme^{s H_{1-s} (A^n | \barf_{X_n}(A^n), X_n |  P_{A}^n  ) } \right] \dotgeq  \exp \left(  ns  | H_{1-s}(A|P_{A}) -  R|^+  \right). \label{eqn:bound_e1}
\end{equation}
For the case $|\calE|=1$,  this establishes~\eqref{eqn:strong_lb} for $s\in [0,s_0(A|P_A)]$.


Now we extend our analysis to   $|\calE|>1$. Naturally, we operate on a type-by-type basis over $\calE^n$. Analogously to the derivation of~\eqref{eqn:bound_e1}  via Lemma \ref{lem:exp}, we see that  if $0\le s\le s_0 = \min_e s_0(A|P_{A|E=e})$, we have 
\begin{align}
&\bbE_{X_n} \left[ \rme^{s H_{1-s} (A^n | \barf_{X_n}(A^n), E^n, X_n | P_{AE}^n  ) } \right] \nn\\*
&\dotgeq\sum_{Q_E\in\calP_n(\calE) } P_E^n (\calT_{Q_E})   \exp \bigg(  ns \bigg |  \sum_{e\in\calE} Q_E(e) H_{1-s}(A|P_{A|E=e}) - R  \bigg|^+ \bigg) . \label{eqn:split_e}
\end{align}
See Remark~\ref{rmk:conditional} in Appendix~\ref{app:exp}  for a detailed description of this step.   
In fact, his derivation is similar to the corresponding calculations in Merhav's work in~\cite[Section~IV-C]{merhav08} and~\cite[Section~IV-D]{merhav14}.   Because $ P_E^n (\calT_{Q_E})  \doteq \rme^{-nD(Q_E \| P_E)}$,  
\begin{align}
&\bbE_{X_n} \left[ \rme^{s H_{1-s} (A^n | \barf_{X_n}(A^n), E^n, X_n | P_{AE}^n  ) } \right] \nn\\*
&\dotgeq  \exp\left( n \max_{Q_E \in \calP(\calE)} \left\{ -D(Q_E \| P_E) +   s \bigg |  \sum_{e\in\calE} Q_E(e) H_{1-s}(A|P_{A|E=e}) - R  \bigg|^+  \right\}\right) . \label{eqn:Ege1}
\end{align} 
Denote the  optimizer in the maximization in \eqref{eqn:Ege1} as $Q_E^*$. If the  $|\fndot|^+$ is inactive for $Q_E^*$, by straightforward calculus, 
\begin{align}
Q_E^*(e)= \frac{P_E(e) \rme^{s H_{1-s} ( A|P_{A|E=e} ) }}{\sum_{e' \in\calE} P_E(e') \rme^{s H_{1-s}( A|P_{A|E=e'} ) }} , \quad\forall\,  e\in \calE,
\end{align}
while if the $|\fndot|^+$ is active for $Q_E^*$, obviously $Q_E^*(e) = P_E(e)$ for all $e\in \calE$.  By using these forms of the optimizer $Q_{E}^*$ and the fact that the conditional R\'enyi entropy $H_{1-s}(A|E|P_{AE})$ (defined in \eqref{eqn:renyi_ent_Q})  is related to the distribution $P_E$ and the unconditional  R\'enyi entropies  $\{H_{1-s}( A|P_{A|E=e }):e\in\calE\}$ as follows 
\begin{equation}
sH_{1-s}(A|E|P_{AE})  =  \log\sum_{  {e\in\calE} }  P_E(e ) \rme^{s H_{1-s}( A|P_{A|E=e } ) },
\end{equation}
  we see that the maximization in \eqref{eqn:Ege1} reduces to $s|H_{1-s}(A|E|P_{AE})-R|^+$.  Upon taking the $\log$, dividing by $ns$, and taking the $\varliminf$, we complete  the proof of \eqref{eqn:strong_lb} for the case where $|\calE|>1$, assuming $s\in [0,s_0]$. As such, we have completed the proof of the lower bound on the optimal compression rate for strongly universal   hash functions   in \eqref{eqn:optkey1_s0}.

\subsubsection{Lower Bound to $\widetilde{T}_{ s}$ }\label{sec:lb2}
For the lower bound to  $\widetilde{T}_{s}$, we  note  from the work by Hayashi in \cite[Lemma 5]{Hayashi_security} that for any $s\in (-1,1)\setminus \{0\}$ that  the conditional R\'enyi entropy and its Gallager form satisfy
\begin{equation}
H_{1+s} (A|E|P_{AE}) \ge H_{\frac{1}{1-s}}^\uparrow(A|E|P_{AE}) .
\end{equation} 
Furthermore,  because $s\mapsto H_{1+s}^\uparrow$ is monotonically non-increasing, 
\begin{align}
&\frac{1}{n}H_{1+s} (A^n|f_n (A^n) , E^n| P_{AE}^n )\nn\\*
& \ge   \frac{1}{n} H_{ \frac{1}{1-s}}^\uparrow(A^n|f_n (A^n) , E^n| P_{AE}^n ) \label{eqn:lb_min10} \\
&= \frac{1 }{n}  H_{ \infty}^\uparrow (A^n|f_n (A^n) , E^n | P_{AE}^n ). \label{eqn:lb_min1}
\end{align}
We require that the term on the leftmost term of this inequality to vanish to since the constraint $\widetilde{G}(R,s)=0$ is present. This implies that  the normalized Gallager min-entropy $\frac{1 }{n}  H_{ \infty}^\uparrow$  necessarily vanishes. From the relation between the strong converse exponent and $H_\infty^\uparrow$ in~\eqref{eqn:str_conv_exp}, we see that $-\frac{1}{n}\log \rmP_{\rmc}^{(n)} (f_n)\to 0$, where  $\rmP_{\rmc}^{(n)}(f_n)$ is the  probability of correct optimal  (MAP) decoding under encoder $f_n$. By the  exponential   strong converse for Slepian-Wolf coding  by Oohama and Han~\cite[Theorem 2]{Oohama94}, we know that if $R< H(A|E|P_{AE})$, it is also necessarily   true that $\varliminf_{n\to\infty}-\frac{1}{n}\log \rmP_{\rmc}^{(n)}(f_n) > 0$. Hence, by contraposition,   $R  \ge H(A|E|P_{AE})$. Thus,  $
\widetilde{T}_{ s} \ge H(A|E|P_{AE})$.   
This, together with the corresponding upper bound proved in Section \ref{sec:ub}, immediately  establishes the lower bound to~\eqref{eqn:optkey1a} for all $s \in (0,1)$.
\end{proof}

\subsection{Proofs of  \eqref{eqn:optkey2} and \eqref{eqn:optkey2a}} \label{sec:prf_optkey2}
\begin{proof}
To prove  \eqref{eqn:optkey2} and \eqref{eqn:optkey2a},  first recall  the definitions of the optimal rates $T_{ s}^\uparrow$ and $ \widetilde{T}_s^\uparrow$ in Theorem \ref{cor:threshold}.  Since $\widetilde{G}^\uparrow(R,s)\le{G}^\uparrow(R, s)$,   and both functions are monotonically non-increasing in $R$,  it holds that $\widetilde{T}^\uparrow_s\le T^\uparrow_s$ for all $s$. 

\subsubsection{Upper Bounds}
The upper bounds for $T_{- s}^\uparrow$ and $T_{ s}^\uparrow $ can be shown in the same way as the   arguments to upper bound $T_{-s}$ and $T_{s}$ in Section \ref{sec:ub} and  using the results in \eqref{eqn:min_gal_res} and \eqref{eqn:plus_max_uni_str}.  Details are omitted for brevity. 

\subsubsection{Lower Bounds} \label{sec:lower_bds}
For the lower bound  to $\widetilde{T}_{ - s}^\uparrow$, we use an important result by Fehr and Berens \cite[Theorem~3]{fehr}, which in our context states that for any hash function $f_n:\calA\to\{1,\ldots, M\}$, we have 
\begin{align}
H_{1-s}^\uparrow (A|f_n(A) , E| P_{AE}) \ge  H_{1-s}^\uparrow (A|  E| P_{AE})-\log M \label{eqn:fehr}
\end{align}
for all $s\in (-\infty, 1)$. Note that $s\in (-\infty,1)$ includes the range  of interest  for $\widetilde{T}_{ - s}^\uparrow$  which is $s\in [0,1/2]$.  The inequality in \eqref{eqn:fehr} is a consequence of   monotonicity under conditioning and the chain rule for the Gallager form of the conditional R\'enyi entropy.  See \eqref{eqn:mono} and \eqref{eqn:cr}.
Since $M = \rme^{nR}$ and $H_{1-s}^\uparrow (A^n|  E^n| P_{AE}^n) = nH_{1-s}^\uparrow (A|  E| P_{AE})$, this  implies that for $s\in [0,1/2]$, we have   
$
\widetilde{G}^{\uparrow}(R,-s) \ge | H_{1-s}^\uparrow(A|E|P_{AE}) - R |^+ $ which immediately leads  to  the bound $\widetilde{T}_{-s}^\uparrow \ge H_{1-s}^\uparrow(A|E|P_{AE})$,  
establishing~\eqref{eqn:optkey2}.  

Now, for the lower bound of $\widetilde{T}_{ s}^\uparrow$ in~\eqref{eqn:optkey2a}, we note from the monotonically decreasing nature of $s\mapsto H_{1+s}^\uparrow$ that  
\begin{equation}
\frac{1}{n} H_{1+s}^\uparrow (A^n|f_n (A^n) , E^n| P_{AE}^n )\ge \frac{1 }{n}  H_{ \infty}^\uparrow (A^n|f_n (A^n) , E^n | P_{AE}^n ) \label{eqn:lb_min}
\end{equation}
for all $s\in [0,\infty)$. 
We require that the term on the left to vanish since the constraint $\widetilde{G}^\uparrow(R,s)=0$ is present. Hence,  $\frac{1 }{n}  H_{ \infty}^\uparrow $ also vanishes. Similarly to the argument in Section \ref{sec:lb2}  after \eqref{eqn:lb_min1},  by invoking the exponential  strong converse to the  Slepian-Wolf theorem \cite[Theorem 2]{Oohama94}, we know that
$
\widetilde{T}_{ s}^\uparrow \ge H(A|E|P_{AE})
$.  
 This    establishes \eqref{eqn:optkey2a}.
\end{proof}
\begin{remark}
We remark that the proof of the lower bound to $\widetilde{T}_{-s}^\uparrow$  above is  much simpler than the proof of the   lower bound  of $\overline{T}_{-s}$  for strongly universal hash functions in~Section~\ref{sec:lb1} because we can leverage two useful properties of  $H_{1-s}^\uparrow $ (monotonicity and chain rule) leading to \eqref{eqn:lb_min}.  In contrast, $H_{1 - s}$ does not possess these properties. Observe that the proof in the first half of Section~\ref{sec:lower_bds} also allows us to conclude that the upper bound  in~\eqref{eqn:min_gal_res} for the direct part is coincident with the lower bound  on $\widetilde{G}^\uparrow(R,-s)$ for all $s\in [0,1/2]$, i.e., 
\begin{equation}
G^\uparrow(R,-s)=\widetilde{G}^\uparrow(R,-s) = | H_{1-s } ^\uparrow(A|E|P_{AE}) - R|^+ . 
\end{equation}

\end{remark}

\section{Proof of Theorem~\ref{thm:exponents}}\label{sec:prf_exp}
We  prove statements \eqref{eqn:exp1} and~\eqref{eqn:exp2} in Subsections \ref{sec:dir_exp_1} and \ref{sec:dir_exp_2} respectively. Statements \eqref{eqn:exp3}  and \eqref{eqn:exp4}   are jointly proved in Subsection  \ref{sec:dir_exp_4}.

\subsection{Proof of  \eqref{eqn:exp1} in Theorem \ref{thm:exponents}}\label{sec:dir_exp_1}
First we note that all the exponents are non-negative since $H_{1-s}(A^n | f_{X_n}(A^n), E^n, X_n | P_{AE}^n \times P_{X_n} )= O(n)$ and similarly for all the other R\'enyi information quantities. This gives the $|\fndot|^+$ signs in all lower bounds in~\eqref{eqn:exp1}--\eqref{eqn:exp4}. 

Fix $t\in [s,1]$. The one-shot bound in \eqref{10-15-30}  in Lemma \ref{lem:one-shot-direct} implies that     for any $\epsilon$-almost universal$_2$ hash function $f_{X_n}$, 
\begin{align}
&-\log\bbE_{X_n}\left[H_{1-s} (A^n | f_{X_n} (A^n) , E^n, X_n | P_{AE}^n ) \right]\nn\\*
&\ge -\log\bbE_{X_n}\left[H_{1-t} (A^n | f_{X_n} (A^n) , E^n, X_n | P_{AE}^n ) \right] \label{eqn:low_s_t}\\*
&=-\log \left\{\frac{1 }{t}\bbE_{X_n}\left[\log\left[\rme^{t H_{1-t} (A^n | f_{X_n} (A^n) , E^n, X_n | P_{AE}^n  ) } \right]\right] \right\}\\
&\ge -\log \left\{\frac{1 }{t} \log\bbE_{X_n} \left[\rme^{ t H_{1-t}  (A^n | f_{X_n} (A^n) , E^n, X_n | P_{AE}^n ) }  \right]\right\} \label{eqn:jens_log}\\
 &\ge -\log \left\{ \frac{1 }{t} \log \left( 1 + \frac{ \epsilon^{t} \rme^{t H_{1-t}  (A^n|E^n|P_{AE}^n )  } }{M^{t}} \right) \right\}\\
 &\ge -\log\left\{ \frac{1 }{t} \cdot  \frac{ \epsilon^{t}  \rme^{ n t H_{1-t}  (A |E  |P_{AE}  )  } }{M^{t}} \right\} \label{eqn:log_ine}\\
 &= -\log\left\{ \frac{1 }{t}\epsilon^{t}\right\} + nt (  R-H_{1-t}  (A |E  |P_{AE}  )  ), \label{eqn:log_ine2}
\end{align}
where in \eqref{eqn:low_s_t} we used the fact that $t\mapsto H_{1-t}$ is monotonically non-decreasing, in \eqref{eqn:jens_log} we applied Jensen's inequality to the concave function $t\mapsto \log t$,   and in \eqref{eqn:log_ine} we employed the bound  $\log(1+ t)\le t$. The   bound  in \eqref{eqn:log_ine2} holds for all $f_{X_n}\in\calU_2$ and all $t\in [s,1 ]$. Now,  we normalize by $n$ and take the $\varliminf$ as $n\to\infty$. Finally, we maximize over all $t\in [s,1]$. This yields~\eqref{eqn:exp1}, concluding the proof.


\subsection{Proof of  \eqref{eqn:exp2} in Theorem \ref{thm:exponents}}\label{sec:dir_exp_2}
Fix $t\in [s,1/2]$. The one-shot bound in \eqref{10-15-31} in Lemma \ref{lem:one-shot-direct}  implies  that  for any $\epsilon$-almost universal$_2$ hash function $f_{X_n}$, 
\begin{align}
&-\log\bbE_{X_n}\left[H_{1-s}^\uparrow (A^n | f_{X_n} (A^n) , E^n, X_n | P_{AE}^n ) \right]\nn\\*
&\ge -\log\bbE_{X_n}\left[H_{1-t}^\uparrow (A^n | f_{X_n} (A^n) , E^n, X_n | P_{AE}^n ) \right]\label{eqn:low_s_t2}\\*
&=-\log \left\{\frac{1-t}{t}\bbE_{X_n}\left[\log\left[\rme^{ \frac{t}{1-t} H_{1-t}^\uparrow (A^n | f_{X_n} (A^n) , E^n, X_n | P_{AE}^n ) } \right]\right] \right\}\\
&\ge -\log \left\{\frac{1-t}{t} \log\bbE_{X_n} \left[\rme^{ \frac{t}{1-t} H_{1-t}^\uparrow (A^n | f_{X_n} (A^n) , E^n, X_n | P_{AE}^n ) }  \right]\right\}\\
 &\ge -\log \left\{ \frac{1-t}{t} \log \left( 1 + \frac{ \epsilon^{\frac{t}{1-t}} \rme^{ \frac{t}{1-t} H_{1-t}^\uparrow (A^n|E^n|P_{AE}^n )  } }{M^{\frac{t}{1-t}}} \right) \right\}\\
 &\ge -\log\left\{ \frac{1-t}{t} \cdot \frac{ \epsilon^{\frac{t}{1-t}} \rme^{ n \frac{t}{1-t} H_{1-t}^\uparrow (A |E  |P_{AE}  )  } }{M^{\frac{t}{1-t}}} \right\}\\
 &= -\log\left\{ \frac{1-t}{t}\epsilon^{\frac{t}{1-t}}\right\} + n\frac{t}{1-t} \left(  R-H_{1-t}^\uparrow (A |E  |P_{AE}  )   \right)
\end{align}
In  \eqref{eqn:low_s_t2}, we used the fact that for $t\mapsto H_{1-t}^\uparrow$ is monotonically non-decreasing.   
This bound holds for all $f_{X_n}\in\calU_2$ and all $t\in [s,  1/2 ]$. Now,  we normalize by $n$ and take the $\varliminf$ as $n\to\infty$. Finally, we maximize over all $t\in [s,1/2]$. This yields \eqref{eqn:exp2}, concluding the proof.


\subsection{Proofs of  \eqref{eqn:exp3} and \eqref{eqn:exp4} in Theorem \ref{thm:exponents}}\label{sec:dir_exp_4}
We prove \eqref{eqn:exp4}  before proving \eqref{eqn:exp3}. We note that 
\begin{align}
&\varliminf_{n\to\infty}-\frac{1}{n}\log H_{1+s }^\uparrow (A^n | f_{X_n} (A^n) , E^n, X_n | P_{AE}^n\times P_{X_n})\nn\\*
&\ge\varliminf_{n\to\infty}-\frac{1}{n}\log H_{1-s'}^\uparrow (A^n | f_{X_n} (A^n) , E^n, X_n | P_{AE}^n\times P_{X_n})
\end{align}
for all $s \in [0,\infty)$ and $s'\in [0,1]$. This is because $\alpha\mapsto H_{\alpha}^\uparrow$ is monotonically non-increasing. This bound implies that $E^\uparrow(R,s )\ge E^\uparrow(R,-s')$. 
 Combining this  with the lower bound   on the exponent of $H_{1-s'}^\uparrow$  in \eqref{eqn:exp2}  and noting that the  lower bound is maximized at $s'=0$ immediately establishes~\eqref{eqn:exp4}.  
 
Finally, we note from \eqref{eqn:ent_min} that  $H_{1+s}\le H_{1+s}^\uparrow$ so 
 \begin{align}
& \varliminf_{n\to\infty}-\frac{1}{n}\log H_{1+s } (A^n | f_{X_n} (A^n) , E^n, X_n | P_{AE}^n\times P_{X_n}) \nn\\*
&\ge   \varliminf_{n\to\infty}-\frac{1}{n}\log H_{1+s }^\uparrow (A^n | f_{X_n} (A^n) , E^n, X_n | P_{AE}^n\times P_{X_n}) .
 \end{align}
 Combining this with the lower bound on $E^\uparrow(R,s)$ in \eqref{eqn:exp4} completes the proof of \eqref{eqn:exp3}.
 
\section{Conclusion and Future Work} \label{sec:con} 

In this paper, we have developed novel techniques to bound the asymptotic behaviors of remaining uncertainties measured according to various conditional R\'enyi entropies. This is in contrast to other works~\cite{Liang,BlochBarros,Wyn75,tandon13,villard13, bross16} that quantify uncertainty using Shannon information measures. We motivated our study by showing that the quantities we characterize are generalizations of the error exponent and the strong converse exponent  for the Slepian-Wolf problem. We studied various important classes of hash functions, including universal$_2$ and strongly universal hash functions. Finally, we also showed that in many cases, the optimal compression rates to ensure that the normalized remaining uncertainties vanish can be  characterized exactly, and that they exhibit behaviors that are somewhat different  to when Shannon information measures are used.

In the future, we hope to derive lower bounds to the normalized remaining uncertainties and upper bounds on their exponents that match or approximately match their achievability counterparts in Theorems \ref{thm:rem} and \ref{thm:exponents}. In addition, just as in the authors' earlier work in~\cite{HayashiTan15}, we may also study the {\em second-order} or $\sqrt{n}$ behavior~\cite{TanBook} of the remaining uncertainties.  These challenging endeavors require  the development of new one-shot bounds as well as the application of new large-deviation and central-limit-type  bounds on various probabilities. 
\appendices

\section{One-Shot Direct Part Bounds} \label{app:one-shot}
 
\newcommand{\rE}{\bbE}
\newcommand{\sM}{M}

In this appendix, we state and prove several one-shot bounds on the various conditional R\'enyi entropies.  Lemmas \ref{lem:one-shot-direct2} and \ref{lem:one-shot3} are used in the proofs for the remaining uncertainties (Theorem~\ref{thm:rem}).  Lemma \ref{lem:one-shot-direct} is used in the proofs for the exponents (Theorem~\ref{thm:exponents}).

\begin{lemma}\label{lem:one-shot-direct2}
For $\epsilon$-almost universal$_2$ hash functions $f_X:{\cal A} \to {\cal M}=\{1, \ldots, \sM\}$,
we have
\begin{align}
\rE_X   \rme^{-s H_{1+s}(A|f_X(A),E,X |P_{AE} )}
& \ge 
 2^{-s}\sum_e P_E(e) \sum_{a:P_{A|E}(a|e) \ge \frac{\epsilon}{\sM}} 
P_{A|E}(a|e)
 \nn\\*
 &\qquad + 
\left(\frac{\epsilon}{\sM} \right)^{-s}
2^{-s}\sum_e P_E(e) \sum_{a:P_{A|E}(a|e) < \frac{\epsilon}{\sM}} 
P_{A|E}(a|e)^{1+s}\label{10-16-4}
\end{align}
for any $s \in [0,1]$.
\end{lemma}

\begin{proof} 
We first establish some basic inequalities: For any $(a,e)\in\calA \times \calE$,   and any $\epsilon$-almost universal$_2$ hash function $f_{X_n}:{\cal A} \to {\cal M}=\{1, \ldots, \sM\}$,
\begin{align}
\bbE_{X} \sum_{a'\in f_X^{-1}(f_X(a))}P_{A|E}(a'|e)&\le    P_{A|E}(a|e ) +\frac{\epsilon}{M}\sum_{a'\ne a}P_{A|E}(a'|e)\label{eqn:expect_hash}\\
&\le    P_{A|E}(a|e ) +\frac{\epsilon}{M}\label{eqn:expect_hash2}\\
&\le  2\max\left\{ P_{A|E}(a |e),\frac{\epsilon}{M} \right\}. \label{10-14-7}
\end{align}
Using \eqref{10-14-7}, we have
\begin{align}
&\rE_X \rme^{-s H_{1+s}(A|f_X(A),E,X|P_{AE})}\nn\\*
& = 
\rE_X \sum_e P_E(e) \sum_{i} \left(\sum_{a\in f_X^{-1}(i) }P_{A|E}(a|e) \right)
\left(
\sum_{a'\in f_X^{-1}(i) }
\left(\frac{P_{A|E}(a'|e)}{ \sum_{a\in f_X^{-1}(i) }P_{A|E}(a|e )}\right )^{1+s}
\right) \\
& =
\rE_X \sum_e P_E(e) \sum_{i}  \left(\sum_{a\in f_X^{-1}(i) }P_{A|E}(a|e) \right)^{-s}
\left(
\sum_{a'\in f_X^{-1}(i) }
P_{A|E}(a'|e)^{1+s}
\right) \\
& =
\rE_X \sum_e P_E(e) \sum_{a} 
P_{A|E}(a|e)^{1+s}
\left(\sum_{a'\in f_X^{-1}(f_X(a)) }P_{A|E}(a'|e)\right)^{-s} \\
& \ge
 \sum_e P_E(e) \sum_{a} 
P_{A|E}(a|e)^{1+s}
\left(\rE_X \sum_{a'\in f_X^{-1}(f_X(a)) }P_{A|E}(a'|e)\right)^{-s} \\
& \ge
 \sum_e P_E(e) \sum_{a} 
P_{A|E}(a|e)^{1+s}
\left(2 \max\left\{ P_{A|E}(a|e)  , \frac{\epsilon}{\sM} \right\}\right)^{-s} \\
& =
 2^{-s}\sum_e P_E(e) \sum_{a} 
P_{A|E}(a|e)^{1+s}
 \max \left\{ P_{A|E}(a|e)^{-s}  , \left(\frac{\epsilon}{\sM}\right)^{-s} \right\} \\
& =
 2^{-s}\sum_e P_E(e) \sum_{a:P_{A|E}(a|e) \ge \frac{\epsilon}{\sM}} 
P_{A|E}(a|e)
+ 
\left(\frac{\epsilon}{\sM}\right)^{-s}
2^{-s}\sum_e P_E(e) \sum_{a:P_{A|E}(a|e) < \frac{\epsilon}{\sM}} 
P_{A|E}(a|e)^{1+s}
\end{align}
Thus, we obtain \eqref{10-16-4}.
\end{proof}

\begin{lemma} \label{lem:one-shot3}
For $\epsilon$-almost universal$_2$ hash functions $f_X:{\cal A} \to {\cal M}=\{1, \ldots, \sM\}$,
we have
\begin{align}
&\rE_X \rme^{-\frac{s}{1+s} H_{1+s}^{\uparrow}(A|f_X(A),E,X|P_{A E})} \ge  
2^{-\frac{s}{1+s}}
 \sum_e P_E(e) 
\sum_{a: P_{A|E}(a|e)^{1+s} \ge \frac{\epsilon}{\sM}
\sum_{a'} P_{A|E}(a'|e)^{1+s}} 
P_{A|E}(a|e) \nn\\
&\qquad  +
2^{-\frac{s}{1+s}}
\left(\frac{\epsilon}{\sM} \right)^{-\frac{s}{1+s}}
\sum_e P_E(e) 
\sum_{a: P_{A|E}(a|e)^{1+s} < \frac{\epsilon}{\sM}
\sum_{a'} P_{A|E}(a'|e)^{1+s}} 
P_{A|E}(a|e)^{1+s}
\left( 
\sum_{a'} P_{A|E}(a'|e)^{1+s}
\right)^{-\frac{s}{1+s}}.
\label{10-16-4b}
\end{align}
for any $s \in [0,\infty)$.
\end{lemma}

\begin{proof} 
Using \eqref{10-14-7}, we have
\begin{align}
& \rE_X \rme^{-s H_{1+s}^{\uparrow}(A|f_X(A),E,X|P_{AE})} \\
& = 
\rE_X \sum_e P_E(e) \sum_{i} \left(\sum_{a\in f_X^{-1}(i) }P_{A|E}(a|e)\right)
\left(
\sum_{a'\in f_X^{-1}(i) }
\left(\frac{P_{A|E}(a'|e)}{ \sum_{a\in f_X^{-1}(i) }P_{A|E}(a|e )} \right)^{1+s}
\right)^{\frac{1}{1+s}} \\
&=
\rE_X \sum_e P_E(e) \sum_{i} 
\left(\sum_{a\in f_X^{-1}(i) }P_{A|E}(a|e)^{1+s} \right)
\left(\sum_{a'\in f_X^{-1}(i)}P_{A|E}(a'|e)^{1+s}\right)^{-\frac{s}{1+s}} \\
& = 
\rE_X \sum_e P_E(e) \sum_{a} 
P_{A|E}(a|e)^{1+s}
\left(\sum_{a'\in f_X^{-1}(f_X(a)) }P_{A|E}(a'|e)^{1+s}\right)^{-\frac{s}{1+s}} \\
&\ge 
 \sum_e P_E(e) \sum_{a} 
P_{A|E}(a|e)^{1+s}
\left(\rE_X \sum_{a'\in f_X^{-1}(f_X(a)) }P_{A|E}(a'|e)^{1+s}\right)^{-\frac{s}{1+s}} \\
&\ge 
 \sum_e P_E(e) \sum_{a} 
P_{A|E}(a|e)^{1+s}
\left(2 
\max \left\{ P_{A|E}(a|e)^{1+s}, 
\frac{\epsilon}{\sM}
\sum_{a'} P_{A|E}(a'|e)^{1+s} \right\}
\right)^{-\frac{s}{1+s}} \\
&\ge 
 \sum_e P_E(e) 
\sum_{a: P_{A|E}(a|e)^{1+s} \ge \frac{\epsilon}{\sM}
\sum_{a'} P_{A|E}(a'|e)^{1+s}} 
P_{A|E}(a|e)^{1+s}
(2 
 P_{A|E}(a|e)^{1+s}
)^{-\frac{s}{1+s}} \nn\\*
& \qquad +
\sum_e P_E(e) 
\sum_{a: P_{A|E}(a|e)^{1+s} < \frac{\epsilon}{\sM}
\sum_{a'} P_{A|E}(a'|e)^{1+s}} 
P_{A|E}(a|e)^{1+s}
\left(2 
\frac{\epsilon}{\sM}
\sum_{a'} P_{A|E}(a'|e)^{1+s}
\right)^{-\frac{s}{1+s}} \\
&= 
2^{-\frac{s}{1+s}}
 \sum_e P_E(e) 
\sum_{a: P_{A|E}(a|e)^{1+s} \ge \frac{\epsilon}{\sM}
\sum_{a'} P_{A|E}(a'|e)^{1+s}} 
P_{A|E}(a|e) \nn\\*
& \qquad+
2^{-\frac{s}{1+s}}
\left(\frac{\epsilon}{\sM} \right)^{-\frac{s}{1+s}}
\sum_e P_E(e) 
\sum_{a: P_{A|E}(a|e)^{1+s} < \frac{\epsilon}{\sM}
\sum_{a'} P_{A|E}(a'|e)^{1+s}} 
P_{A|E}(a|e)^{1+s}
\left( 
\sum_{a'} P_{A|E}(a'|e)^{1+s}
\right)^{-\frac{s}{1+s}}.
\end{align}
Thus, we obtain \eqref{10-16-4b}.
\end{proof}

\begin{lemma} \label{lem:one-shot-direct}
For $\epsilon$-almost universal$_2$ hash functions $f_X:{\cal A} \to {\cal M}=\{1, \ldots, \sM\}$,
we have
\begin{align}
\rE_X \rme^{s H_{1-s}(A|f_X(A),E,X |P_{AE}  )}
& \le 1+ 
\frac{\epsilon^s \rme^{s H_{1-s}(A|E|P_{AE})}}{\sM^s}, 
\label{10-15-30}
\end{align} for all $s \in [0,1]$. In addition, we also have 
\begin{align}
\rE_X \rme^{\frac{s}{1-s} H_{1-s}^{\uparrow}(A|f_X(A),E,X |P_{AE}  )}
& \le 
1+ 
\frac{\epsilon^{\frac{s}{1-s}}
\rme^{\frac{s}{1-s} H_{1-s}^{\uparrow}(A|E|P_{AE})} 
}{\sM^{\frac{s}{1-s}}}, 
\label{10-15-31}
\end{align}
for all $s \in  \left[0,1/2 \right]$.
\end{lemma}

\begin{proof}
We have 
\begin{align}
&\rE_X \rme^{s H_{1-s}(A|f_X(A),E,X|P_{AE})} \nn\\*
& =
\rE_X \sum_e P_E(e) \sum_{i} \left(\sum_{a\in f_X^{-1}(i) }P_{A|E}(a|e) \right)
\left(
\sum_{a'\in f_X^{-1}(i) }
\left(\frac{P_{A|E}(a'|e)}{ \sum_{a\in f_X^{-1}(i) }P_{A|E}(a|e) } \right)^{1-s}
\right) \\
& =
\rE_X \sum_e P_E(e) \sum_{i} \left(\sum_{a\in f_X^{-1}(i) }P_{A|E}(a|e)\right)^s
\left(
\sum_{a'\in f_X^{-1}(i) }
P_{A|E}(a'|e)^{1-s}
\right) \\
& =
\rE_X \sum_e P_E(e) \sum_{a} 
P_{A|E}(a|e)^{1-s}
\left(\sum_{a'\in f_X^{-1}(f_X(a)) }P_{A|E}(a'|e) \right)^s \label{eqn:derivation_H_1minuss} \\
& \le
 \sum_e P_E(e) \sum_{a} 
P_{A|E}(a|e)^{1-s}
\left(\rE_X \sum_{a'\in f_X^{-1}(f_X(a)) }P_{A|E}(a'|e)\right)^s \label{eqn:jensen1} \\
& \le
 \sum_e P_E(e) \sum_{a} 
P_{A|E}(a|e)^{1-s}
\left(P_{A|E}(a|e)  + \frac{\epsilon}{\sM}\right)^s  \label{eqn:use_this_fact}\\
& \le
 \sum_e P_E(e) \sum_{a} 
P_{A|E}(a|e)^{1-s}
\left(P_{A|E}(a|e)^s  +  \left(\frac{\epsilon}{\sM}  \right)^s  \right) \label{eqn:use_this_fact2}\\
& =
 \sum_e P_E(e) \sum_{a} 
\left(P_{A|E}(a|e)
+P_{A|E}(a|e)^{1-s} \left(\frac{\epsilon}{\sM} \right)^s \right) \\
& =
1+ \sum_e P_E(e) \sum_{a} P_{A|E}(a|e)^{1-s} \left(\frac{\epsilon}{\sM} \right)^s  \\
&=
1+ 
\frac{\epsilon^s \rme^{s H_{1-s}(A|E|P_{AE})}}{\sM^s}.
\end{align} In \eqref{eqn:jensen1}, we   applied Jensen's inequality with the concave function $t\mapsto t^{s}$. Here is where the condition $s \in [0,1 ]$ is used.  In~\eqref{eqn:use_this_fact}, we used \eqref{eqn:expect_hash2} and in \eqref{eqn:use_this_fact2}, we used the fact that $(b+c)^s\le b^s+c^s$ for $b,c\ge 0 $ and $s\in [0,1]$  \cite[Problem 4.15(f)]{gallagerIT}.  
Thus, we obtain \eqref{10-15-30}.

For \eqref{10-15-31}, consider,
\begin{align}
&\rE_X \rme^{\frac{s}{1-s} H_{1-s}^{\uparrow}(A|f_X(A),E,X|P_{AE})}\nn\\*
& =
\rE_X \sum_e P_E(e) \sum_{i}  \left(\sum_{a\in f_X^{-1}(i) }P_{A|E}(a|e) \right)
\left(\sum_{a'\in f_X^{-1}(i) }
\left(\frac{P_{A|E}(a'|e)}{ \sum_{a\in f_X^{-1}(i) }P_{A|E}(a|e )}\right)^{1-s}
\right)^{\frac{1}{1-s}}
 \\
& =
\rE_X \sum_e P_E(e) \sum_{i} 
\left(
\sum_{a'\in f_X^{-1}(i) }
P_{A|E}(a'|e)^{1-s}
\right)^{\frac{1}{1-s}} \\
& =
\rE_X \sum_e P_E(e) \sum_{i} 
\left(
\sum_{a\in f_X^{-1}(i) }
P_{A|E}(a|e)^{1-s}
\right)
\left(
\sum_{a'\in f_X^{-1}(i) }
P_{A|E}(a'|e)^{1-s}
\right)^{\frac{s}{1-s}} 
\\
& =
\rE_X \sum_e P_E(e) 
\sum_{a}
P_{A|E}(a|e)^{1-s}
\left(
\sum_{a'\in f_X^{-1}(f_X(a)) }
P_{A|E}(a'|e)^{1-s}
\right)^{\frac{s}{1-s}} 
\\
& \le 
\sum_e P_E(e) 
\sum_{a}
P_{A|E}(a|e)^{1-s}
\left(
\rE_X \sum_{a'\in f_X^{-1}(f_X(a)) }
P_{A|E}(a'|e)^{1-s}
\right)^{\frac{s}{1-s}}  \label{eqn:jensen2}
\\
& \le 
\sum_e P_E(e) 
\sum_{a}
P_{A|E}(a|e)^{1-s}
\left(P_{A|E}(a|e)^{1-s}  
+ \frac{\epsilon}{\sM}  \sum_{a'\neq a } P_{A|E}(a'|e)^{1-s}
\right)^{\frac{s}{1-s}} 
\\
& \le 
\sum_e P_E(e) 
\sum_{a}
P_{A|E}(a|e)^{1-s}
\left(P_{A|E}(a|e)^{s}  
+ \left(\frac{\epsilon}{\sM} \sum_{a'\neq a } P_{A|E}(a'|e)^{1-s}
\right)^{\frac{s}{1-s}} 
\right)
\\
& \le 
\sum_e P_E(e) 
\sum_{a}
P_{A|E}(a|e)^{1-s}
\left(P_{A|E}(a|e)^{s}  
+ \left(\frac{\epsilon}{\sM} \sum_{a'} P_{A|E}(a'|e)^{1-s}
\right)^{\frac{s}{1-s}} 
\right)
\\
& =
1+
\sum_e P_E(e) 
\sum_{a}
P_{A|E}(a|e)^{1-s}
\left(\frac{\epsilon}{\sM} \sum_{a'} P_{A|E}(a'|e)^{1-s}
\right)^{\frac{s}{1-s}} 
\\
& =
1+
\left(\frac{\epsilon}{\sM} \right)^{\frac{s}{1-s}} 
\sum_e P_E(e) 
\left(\sum_{a}
P_{A|E}(a|e)^{1-s}\right)^{1+\frac{s}{1-s}} 
\\
& =
1+
\left(\frac{\epsilon}{\sM} \right)^{\frac{s}{1-s}} 
\sum_e P_E(e) 
\left(\sum_{a}
P_{A|E}(a|e)^{1-s}\right)^{\frac{1}{1-s}} 
\\
&=
1+ 
\frac{\epsilon^{\frac{s}{1-s}}
\rme^{\frac{s}{1-s} H_{1-s}^{\uparrow}(A|E|P_{AE})} 
}{\sM^{\frac{s}{1-s}}}.
\end{align}
In \eqref{eqn:jensen2}, we applied Jensen's inequality with the concave function $t\mapsto t^{\frac{s }{1-s}}$. Here is where the condition $s \in [0,1/2]$ is used. The explanations for the other bounds parallel those for the proof of  \eqref{10-15-30} and are omitted for the sake of brevity.
This completes the proof of \eqref{10-15-31}.
\end{proof}

\section{Bound on the Sum of Expectations in  \eqref{eqn:split_typ3}} \label{app:exp}
Recall the definitions of $\gamma(t) = tH_{1+t}(A|P_A)$ and the tilted distribution $P_A^{(t)}(a) = P_A(a)^{1+t}\rme^{\gamma(t)}$ introduced in Section~\ref{sec:opt_rates}.
 It is easily seen (cf.~Section~\ref{sec:info_measures}) that $\gamma(t)$ is strictly concave for $t>-1$. It also holds that
\begin{align}
D(P_A^{(t)} \| P_A) & = \gamma(t) - t \gamma'(t),\quad\mbox{and} \label{eqn:div_t} \\
H(P_A^{(t)}  ) & = (1+t)\gamma'(t) -   \gamma (t) \label{eqn:ent_t} .
\end{align} 
A fact we   use in the sequel is that if $t >  -1$,  the function $t\mapsto H(P_A^{(t)})$ is monotonically decreasing; this can be seen by considering the derivative $\frac{\rmd }{\rmd t}H(P_A^{(t)}) = (1+t)\gamma''(t) < 0$.  
\begin{lemma} \label{lem:exp}
Let the  sum of expectations in  \eqref{eqn:split_typ3} be denoted as 
\begin{align}
\Psi_n :=\sum_{Q \in\calP_n(\calA)} \bbE_{X_n}   \left\{ \bigg[  \sum_{ \tba \in\calT_Q\setminus\{\ba\}  : \barf_{X_n} (\ba)=\barf_{X_n} (\tba)  }  P_A^n (\tba)\bigg]^s  \right\}  .
\end{align}
 Let $t_R\in [-1,\infty)$ be the unique  number satisfying $H(P_A^{(t_R)}) = R \le \log|\calA|$. Then,
\begin{align}
\Psi_n\dotgeq  \exp( -n \Lambda (s,R))\label{eqn:exp_bd}
\end{align}
where 
\begin{align}
\Lambda (s,R) := \left\{  \begin{array}{cc}
 s  ( R + D(P^{ (t_R) }_A \| P_A)) & s-1\le t_R \\
 R + \gamma(s-1)    & s-1  > t_R 
\end{array} \right. . \label{eqn:Gamma}
\end{align}
Furthermore, if $R< H_{1-s}(A|P_A)$, then for all $s\in [0,s_0(A|P_A)]$ (cf.\ definition in \eqref{eqn:def_s0A}),
\begin{align}
\Psi_n &\dotgeq \exp( -n [s (R+D(P_A^{(t_R)} \| P_A) ) ] ) \label{eqn:case1_dom} \\
&\doteq \sum_{Q \in\calP_n(\calA)} \left\{\bbE_{X_n }\bigg[\sum_{ \tba \in\calT_Q\setminus\{\ba\}  : \barf_{X_n} (\ba)=\barf_{X_n} (\tba)  } P_A^n(\tba) \bigg] \right\}^{s} \label{eqn:case1_dom1}  .
\end{align}
The inequality in \eqref{eqn:case1_dom} implies that  under the stated conditions on $R$ and $s$, the first clause in \eqref{eqn:Gamma} is active. 
\end{lemma}
The calculations here are somewhat similar to those in Merhav's work in \cite[Section~IV-C]{merhav08} and \cite[Section~IV-D]{merhav14}  but the whole proof  for the case in which $\calE=\emptyset$ is  included for completeness.  See Remark~\ref{rmk:conditional} for a sketch of  how to extend the analysis to the memoryless but non-stationary case in which $\calE\ne\emptyset$.

We remark that when $s-1=t_R$, 
\begin{equation}
\Lambda(s,R) \big|_{s= 1+t_R} = (1+t_R  )  \gamma'(t_R) .
\end{equation}
By using $H(P_A^{(t_R)})=R$, \eqref{eqn:div_t}, and~\eqref{eqn:ent_t}, we immediately see that  this ``boundary'' case coincides with the two cases of~\eqref{eqn:Gamma} so $s\mapsto\Lambda(s,R)$ is continuous at $1+t_R$. 
 
\begin{proof}
Let $\calG_R := \{ Q\in\calP (\calA): H(Q) > R\}$, $\calG_{R,n} := \calG_R \cap\calP_n(\calA)$ and $\mathrm{cl}(\calG_R)$ be the closure of $\calG_R$.
We   split $\Psi_n$ into the following two sums
\begin{align}
\Psi_n&= \underbrace{ \sum_{Q  \in \calG_{R,n}} \bbE_{X_n}  \left\{ \bigg[  \sum_{ \tba \in\calT_Q\setminus\{\ba\}  : \barf_{X_n} (\ba)=\barf_{X_n} (\tba)  }  P_A^n (\tba)\bigg]^s \right\}  }_{\alpha_n}  \nn\\*
&\qquad  + \underbrace{ \sum_{Q  \in \calG_{R,n}^c} \bbE_{X_n}  \left\{ \bigg[  \sum_{ \tba \in\calT_Q\setminus\{\ba\}  : \barf_{X_n} (\ba)=\barf_{X_n} (\tba)  }  P_A^n (\tba)\bigg]^s \right\} }_{\beta_n}.
\end{align}
  Define $N_Q   :=\sum_{\tba \in\calT_Q \setminus\{\ba\}} \bone \{\barf_{X_n} (\ba)=\barf_{X_n} (\tba) \} $. This is a sum of $L_Q :=| \calT_Q\setminus\{\ba\}  |$ independent and identically distributed $\{0,1\}$-random variables $\{Y_i\}_{i=1}^L$ with   $\Pr(Y_i = 1) =  \rme^{-nR}=:p$ (property of strong universal$_2$ hash functions). Let $\ba_Q \in \calT_Q$ be any generic vector of type $Q$. Let us now lower bound $\alpha_n$ and $\beta_n$.
\begin{itemize}
\item  For the expectation within $\alpha_n$, we have 
\begin{align}
&\bbE_{X_n } \left\{\bigg[ \sum_{ \tba \in\calT_Q\setminus\{\ba\}  : \barf_{X_n} (\ba)=\barf_{X_n} (\tba)  } P_A^n(\tba) \bigg]^{s} \right\} \nn\\*
&=P_A^n(\ba_Q)^{ s}    \bbE_{X_n } \left\{ \bigg[ \sum_{ \tba \in\calT_Q\setminus\{\ba\}  : \barf_{X_n} (\ba)=\barf_{X_n} (\tba)  }  1\bigg]^{s} \right\} \label{eqn:expected_N0}\\
&= P_A^n(\ba_Q)^{ s} \bbE  \left[ N_Q ^{s} \right]  \label{eqn:expected_N}\\
&\dotgeq P_A^n(\ba_Q)^{ s}  \left\{ \bbE  \left[ N_Q \right] \right\}^s  \label{eqn:expected_N2} \\
&= \left\{\bbE_{X_n }\bigg[\sum_{ \tba \in\calT_Q\setminus\{\ba\}  : \barf_{X_n} (\ba)=\barf_{X_n} (\tba)  } P_A^n(\tba) \bigg] \right\}^{s}, \label{eqn:expected_N3} 
\end{align}
where \eqref{eqn:expected_N0} follows  because all $\tba$ in the sum have the same type $Q$, \eqref{eqn:expected_N} from the definition of $N_Q $,  \eqref{eqn:expected_N2} follows from Lemma \ref{lem:conc}  (in Appendix \ref{app:conc}) under the condition that $Q \in \calG_{R,n}$ (so $L_Q\cdot  p =\bbE[N_Q ] \ge (n+1)^{- |\calA|} \exp( n [H(Q)-R]) \to\infty$ and \eqref{eqn:pL_inf} applies).
Thus, we conclude that 
\begin{align}
\alpha_n &\dotgeq   \sum_{Q\in\calG_{R,n}}  (\bbE[N_Q ] )^s P_{A}^n(\ba_Q)^s  \\
&\doteq \max_{Q \in  \mathrm{cl}(\calG_R) } (\bbE[N_Q ] )^s P_{A}^n(\ba_Q)^s .
\end{align}
By further using the fact that  $P_{A}^n(\ba_Q) = \exp( - n [ D(Q \| P_A) + H(Q)])$ \cite[Lemma~2.6]{Csi97}, we have  
\begin{equation}
\alpha_n \dotgeq \exp(-nsR) \exp\Big(-ns \min_{Q  \in \mathrm{cl}(\calG_R)  }    D(Q\| P_A)  \Big) =: \widetilde{\alpha}_n. \label{eqn:Aunder}
\end{equation}
\item Next, we lower bound the expectation in $\beta_n$. The step from \eqref{eqn:expected_N0} to~\eqref{eqn:expected_N} remains the same but we bound $\bbE [ N_Q ^s]$ differently. We have 
\begin{align}
\bbE \left[N_Q ^{s} \right] &\ge 1^{ s } \Pr(N_Q =1) \\
& = 1^{s} \binom{L_Q}{1} p^1 (1-p)^{ L_Q-1} \\
& =  L_Q \cdot p \cdot \exp \big( (L_Q-1)\log (1-p) \big) \\
& \ge  L_Q \cdot p\cdot \exp \bigg( -(L_Q-1)  \frac{ p}{1-p} \bigg) \label{eqn:kl_ineq} \\
& \dotgeq  L_Q \cdot  p  =\bbE [N_Q ]\label{eqn:kl_ineq2} , 
\end{align}
where \eqref{eqn:kl_ineq} follows from the basic inequality $\log (1-t) \ge -\frac{t}{1-t}$ and  \eqref{eqn:kl_ineq2}
 follows from the fact that $(L_Q-1)   p \le  \rme^{-nR}   |\calT_Q \setminus\{\ba\} |\le \exp( n [H(Q)-R])\le 1$ when $Q \in\calG_{R,n}^c = \{ Q \in\calP_n(\calA) : H(Q)\le R\}$. 
Thus, we have 
\begin{align}
\beta_n &\dotgeq  \sum_{Q \in \calG_{R,n}^c }  \bbE[N_Q ]  P_{A}^n(\ba_Q)^s \\
&\doteq \max_{Q \in \calG_R^c }  \bbE[N_Q ]  P_{A}^n(\ba_Q)^s\\
& \doteq \exp( -nR) \exp\Big( - n \min_{Q  \in \calG_R^c } \big[ sD(Q\| P_A) -(1-s)H(Q)\big] \Big) =: \widetilde{\beta}_n. \label{eqn:Bunder}
\end{align}
\end{itemize}
We remark that the evaluations in \eqref{eqn:expected_N2}  and \eqref{eqn:kl_ineq2} are, in fact, exponentially tight\footnote{The intuition here is that if $H(Q)>R$, the random variable $N_Q $ concentrates doubly-exponentially fast to its expectation, which itself is exponentially large. On the other hand, if $H(Q) < R$, $N_Q $ is typically exponentially small, so $\bbE [ N_Q  ]$ is dominated by the term $1^{s} \Pr(N_Q =1)$.  }~\cite[Eqn.~(34)]{merhav08}. This implies that  $\alpha_n \doteq \widetilde{\alpha}_n$ and  $\beta_n \doteq \widetilde{\beta}_n$. However,  we only require the lower bounds.

It is easy to see that the optimal distribution $Q^*$ in the optimization in the exponent of $\widetilde{\alpha}_n$  satisfies $H(Q^*)=R$ (i.e., $Q^*$ lies on the boundary of $\calG_R$). In fact, the exponent (which is the lossless source coding error exponent~\cite[Theorem~2.15]{Csi97}) can be expressed as $D(P_A^{(t_R)} \| P_A)$ where $t_R\ge -1$ is chosen such that $H(P_A^{(t_R)}) = (1+t_R) \gamma'(t_R) -\gamma(t_R)=R$. Thus,
\begin{align}
\widetilde{\alpha}_n&=\exp(-nsR) \exp \big(-ns D\big(P_A^{(t_R)}\| P_A\big)\big)  \\
&=\exp(-ns [R + \gamma(t_R) - t_R\gamma'(t_R) ] ) . \label{eqn:A_asymp}
\end{align}

Now, it is easy to verify (see Shayevitz~\cite[Section IV-A.8]{Shayevitz10} for example) by differentiating the convex function  $g(Q):= sD(Q\| P_A)- (1-s) H(Q)$ that the {\em unconstrained} minimum in the exponent in $\widetilde{\beta}_n$ takes the form  of a tilting of $P_A$, i.e., 
\begin{equation}
Q^*(a) := \frac{P_A^s(a)}{\sum_{a'} P_A^s(a')}  = P_A^{(s-1)}(a),\quad\forall\, a\in\calA. \label{eqn:Q_star}
\end{equation}
Now depending on the value of $s$, we have two different scenarios. First, if $s-1 > t_R$ or  equivalently, $Q^* \in \calG_R^c$, then $\widetilde{\beta}_n\doteq \exp(-n[ R + g(Q^*) ])  = \exp(-n [R +\gamma(s-1)] )$ (using \eqref{eqn:div_t} and \eqref{eqn:ent_t}). On the other hand, if $s-1\le t_R$, the optimal solution in the optimization in  $\widetilde{\beta}_n$ is again attained at the boundary of $\calG_R$ and $\calG_R^c$ (i.e., the constraint $Q\in\calG_R^c$ is active). Thus, $\widetilde{\beta}_n\doteq \widetilde{\alpha}_n$ where $\widetilde{\alpha}_n$ is in~\eqref{eqn:A_asymp}. In summary,  $\widetilde{\beta}_n\doteq \exp( -n \Lambda (s,R))$ where  $\Lambda(s,R)$ is defined in~\eqref{eqn:Gamma}.  Now, clearly $\widetilde{\beta}_n$ always dominates $\widetilde{\alpha}_n$ (i.e., $\widetilde{\beta}_n$ is exponentially  at least  as large as $\widetilde{\alpha}_n$). This is because when $s-1\le t_R$, they are the same, and when $s-1 > t_R$, we are taking an {\em unconstrained} minimum of  $g(Q)$  in the exponent, making the overall expression larger.   We thus obtain the conclusion in \eqref{eqn:exp_bd}.

For the statements in \eqref{eqn:case1_dom}--\eqref{eqn:case1_dom1}, we assume that $R <  H_{1-s}(A|P_A)$ and   $s\in [0,s_0(A|P_A)]$. We claim that these imply that $s-1 \le t_R$, i.e., the first clause in \eqref{eqn:Gamma} is active. Note  from the definition of $s_0(A|P_A)$ in  \eqref{eqn:def_s0A} that $s\le s_0(A|P_A)$ means that 
\begin{equation}
H_{1-s}(A|P_A) \le  H(A|P_A^{(s-1)}) . \label{eqn:compare_H}
\end{equation}
Since  $R< H_{1-s}(A|P_A)$, it holds that $R< H(A|P_A^{(s-1)})$, but this in turn implies that $s-1\le t_R$ because $t\mapsto H(A|P_A^{(t)})$ is  monotonically non-increasing. Thus~\eqref{eqn:case1_dom} holds.

That \eqref{eqn:case1_dom} is exponentially equal to~\eqref{eqn:case1_dom1} follows from the fact that when $R<H_{1-s}(A|P_A)$ and   $s\in [0,s_0(A|P_A)]$,  $\widetilde{\beta}_n$ is of the same exponential order $\widetilde{\alpha}_n$ and the latter is lower bounded (on the exponential scale) by \eqref{eqn:expected_N3} . 
\end{proof}

\begin{remark} \label{rmk:conditional}
To derive a conditional version of Lemma \ref{lem:exp} to obtain  \eqref{eqn:split_e}, we assume that the type of $\be\in\calE^n$ is $Q_E\in\calP_n(\calE)$. The above derivations go through essentially unchanged by averaging with respect to $Q_E$ everywhere. Specifically, we consider $\ba$ and $\tba$ to belong to various ``$V_{A|E}$-shells'' $\calT_{V_{A|E}}(\be) := \{\ba\in\calA^n: (\ba,\be) \in \calT_{Q_E \times V_{A|E}}\}$~\cite[Ch.~2]{Csi97}. The entropies $H(Q)$ are replaced by conditional entropies $H(V_{A|E} |Q_E)$, the relative entropies $D( Q\| P_A)$ by conditional relative entropies $D(V_{A|E} \| P_{A|E} | Q_E)$ and the tilted conditional distribution   is defined as $P_{A|E}^{(t)}(a|e) \propto P_{A|E}(a|e)^{1+t}$ and so on. For the analogues of  \eqref{eqn:case1_dom}--\eqref{eqn:case1_dom1}  to hold, we firstly require $R < \sum_e Q_E(e ) H_{1-s}(A|P_{A | E=e})$. If we further assume that $0\le s\le  s_0 = \min_e s_0(A|P_{A|E=e})$, $H_{1-s}(A|P_{A|E=e})\le H(A|P_{A|E}^{(s-1)}(\cdot | e))$ for all $e\in\calE$,  and so   $R< \sum_e Q_E(e ) H(A|P_{A|E}^{(s-1)}(\cdot |e))$ giving the first clause    in the conditional analogue of~\eqref{eqn:Gamma}, i.e., that $\Lambda(s,R) = s(R + D(P_{A|E}^{(t_R)}  \| P_{A|E} | Q_E))$ where $t_R \ge -1$ satisfies $H(P_{A|E}^{(t_R)} | Q_E) = R$.    These observations yield \eqref{eqn:split_e} upon averaging over all types on $\calE$.
\end{remark}

\section{A Useful Concentration Bound} \label{app:conc}
The following lemma is  essentially a restatement of Lemma~10 in \cite{HayashiTan2015a}.
\begin{lemma} \label{lem:conc}
Let $Y_1, \ldots, Y_L$ be independent random variables, each taking values in $\{0,1\}$ such that $\Pr(Y_i = 1) = p$. Let $N :=\sum_{i=1}^L Y_i$.  For every $s \in [0,1]$ and any $0<\epsilon <1$, 
\begin{equation}
\bbE [ N^s] \ge \lfloor Lp(1-\epsilon) \rfloor^s \left[ 1- \exp\bigg( - L \, \frac{p}{2}\, \epsilon^2 \bigg)\right]. \label{eqn:conc_lower}
\end{equation}
\end{lemma}

In particular, if $Lp$ is a sequence in $n$ that tends to infinity (as $n$ tends to infinity) exponentially fast, then by taking $\epsilon=1/2$  (say) in \eqref{eqn:conc_lower},
\begin{equation}
\bbE [N^s]\dotgeq (Lp)^s = \left\{\bbE[ N ]\right\}^s.  \label{eqn:pL_inf}
\end{equation}
In fact, we have $\bbE[N^s ] \doteq \left\{\bbE[ N ]\right\}^s$ because $\bbE[N^s]\le \left\{\bbE[ N ]\right\}^s$ by Jensen's inequality.

\subsection*{Acknowledgements}   

VYFT is partially supported an NUS Young Investigator Award (R-263-000-B37-133) and a Singapore Ministry of Education Tier 2 grant ``Network Communication with Synchronization Errors: Fundamental Limits and Codes'' (R-263-000-B61-112).

MH is partially supported by a MEXT Grant-in-Aid for Scientific Research (A) No.\ 23246071. 
MH is also partially supported by the Okawa Reserach Grant
and Kayamori Foundation of Informational Science Advancement. The Centre for Quantum Technologies is funded by the Singapore Ministry of Education and the National Research Foundation as part of the Research Centres of Excellence programme.

\bibliographystyle{unsrt}
\bibliography{isitbib}

\end{document}